\documentclass[11pt]{article}

\usepackage{amssymb,amsmath,amsthm}
\usepackage{enumerate}
\usepackage{epsfig}
\usepackage{version}
\usepackage{fullpage}

\usepackage{colortbl}

\newtheorem{theorem}{Theorem}[section]

\newtheorem{observation}[theorem]{Observation}

\newtheorem{definition}[theorem]{Definition}

\newtheorem{claim}[theorem]{Claim}
\newtheorem{lemma}[theorem]{Lemma}
\newtheorem{polyrule}{Rule}[section]

\def\ie{{\em i.e.}~}

\newenvironment{proofclaim}{
	\noindent \emph{Proof.}
}{%
	\hfill $\diamond$ \\
}

\newcommand{\FAST}{\textsc{$k$-FAST}}
\newcommand{\BIT}{\textsc{$k$-dense-BTI}}

\newcommand{\RTI}{\textsc{$k$-dense-RTI}}
\newcommand{\FASBT}{$k$-FASBT}
\newcommand{\FAS}{$k$-FAS}

\newcommand{\<}{<_{\sigma}}

\begin{document}

\title{Conflict Packing: an unifying technique to obtain polynomial\\ kernels for editing problems on dense instances}

\author{Christophe Paul\thanks{LIRMM, CNRS - Universit\'e Montpellier 2 - \texttt{christophe.paul@lirmm.fr}} \and Anthony Perez\thanks{Univ. Orl\'eans, INSA Centre Val de Loire, LIFO EA 4022 - \texttt{anthony.perez@univ-orleans.fr}} \and St\'ephan Thomass\'e\thanks{LIP - ENS Lyon - \texttt{stephan.thomasse@ens-lyon.fr} }}

\maketitle

\begin{abstract}
We develop a technique that we call \emph{Conflict Packing} in the context of kernelization~\cite{Bod09}, obtaining (and improving) several polynomial kernels for editing problems on dense instances. We apply this technique on several well-studied problems: \textsc{Feedback Arc Set in (Bipartite) Tournaments}, \textsc{Dense Rooted Triplet Inconsistency} and \textsc{Betweenness in Tournaments}. For the former, one is given a (bipartite) tournament $T = (V,A)$ and seeks a set of at most $k$ arcs whose reversal in $T$ results in an acyclic (bipartite) tournament. While a linear vertex-kernel is already known for the first problem~\cite{BFG+11},  using the Conflict Packing allows us to find a so-called \emph{safe partition}, the central tool of the kernelization algorithm in~\cite{BFG+11}, with simpler arguments. For the case of bipartite tournaments, the same technique allows us to obtain a quadratic vertex-kernel. Again, such a kernel was already known to exist~\cite{GX13}, using the concept of so-called \emph{bimodules}. We believe however that our algorithm can provide a better insight on the structural properties of the problem. Regarding \textsc{Dense Rooted Triplet Inconsistency}, one is given a set of vertices $V$ and a \emph{dense} collection $\mathcal{R}$ of rooted binary trees over three vertices of $V$ and seeks a rooted tree over $V$ containing all but at most $k$ triplets from $\mathcal{R}$. As a main consequence of our technique, we prove that the \textsc{Dense Rooted Triplet Inconsistency} problem admits a linear vertex-kernel. This result improves the best known bound of $O(k^ 2)$ vertices for this problem~\cite{GM10}. Finally, we use this technique to obtain a linear vertex-kernel for \textsc{Betweenness in Tournaments}, where one is given a set of vertices $V$ and a dense collection $\mathcal{R}$ of so-called \emph{betweenness triplets} and seeks a linear ordering of the vertices containing all but at most $k$ triplets from $\mathcal{R}$. We would like to mention that this result has recently been improved~\cite{Per13}, based on some of the results we present here. 
\end{abstract}

\section{Introduction}


The theory of \emph{fixed parameter algorithms}~\cite{DF99} has been developed to cope with NP-Hard problems. The goal is to identify a secondary parameter $k$, independent from the data-size $n$, which captures the exponential growth of the complexity cost to solve the problem at hand. That is the complexity of a \emph{fixed parameter tractable} algorithm is $f(k) \cdot n^{O(1)}$, where $f$ is an arbitrary computable function.

As one of the main techniques to design efficient fixed parameter algorithms, \emph{kernelization algorithms}~\cite{Bod09} have attracted a lot of attention during the last few years. A kernelization algorithm transforms in polynomial time, by means of  \emph{reduction rules}, a parameterized instance $(I,k)$ of a parameterized problem into an equivalent instance $(I',k')$, so-called \emph{kernel}, with the property that the parameter $k'$ and the size $|I'|$ of the kernel only depend on $k$. The smaller the size of the kernel is, the faster the problem can be solved. Indeed, the problem can be efficiently tackled on the kernel with any exact algorithm (e.g. bounded search tree or exponential time algorithms). It is well-known that fixed parameterized tractability of a problem is equivalent to the existence of a kernel, but of exponential size only. The standard question for a fixed parameter tractable problem is then to provide a polynomial size kernel or give some evidence of non-existence of such kernel~\cite{BDFH09,BJK11,BTY11}.


The algorithmic tool-kit to design kernelization algorithms is getting richer and richer over the years (the interested reader should refer to standard books \cite{DF99,FG06,Nie06}). In this paper, we systematically apply a kernelization technique, that we called \emph{Conflict Packing}, on four classic parameterized problems: 
\begin{itemize}
\item \textsc{Feedback Arc Set} in Tournaments (\FAST{}) and in Bipartite Tournaments (\FASBT{}): A \emph{feedback arc set} of a digraph $D=(V,A)$ is a subset $F\subseteq A$ of arcs such that reversing the orientation of the arcs of $F$ yields an acyclic digraph. The parameterized \FAST{} problem consists in given a tournament and a parameter $k\in\mathbb{N}$, computing a feedback arc set of size at most $k$ if one exists. The parameterized \FASBT{} problem is the same problem applied on bipartite tournaments.
\item \textsc{Dense Rooted Triplets Inconsistency} (\RTI{}): A \emph{rooted triplet} is a rooted binary tree over three leaves. A
 collection $\mathcal{R}$ of rooted triplets over a leaf set $V$ is \emph{consistent} if there exists a binary tree $T$ over the leaf set $V$ such that every rooted triplet $\{a,b,c\}\in\mathcal{R}$ is homeomorphic to the subtree of $T$ spanning $\{a,b,c\}$. The set $\mathcal{R}$ is \emph{dense} if it contains a rooted triplet for every distinct triplet $\{a,b,c\}\subseteq V$. Given a \emph{dense} set $\mathcal{R}$ of rooted triplets over $V$ and a parameter $k\in\mathbb{N}$, the parameterized \RTI{} problem consists in computing a subset $\mathcal{F}\subseteq \mathcal{R}$ of at most $k$ rooted triplets to edit in order to obtain a consistent set of rooted triplets.

\item \textsc{Dense Betweenness Triples Inconsistency} (\BIT{}): Given a finite set $V$ of vertices and a collection $\mathcal{B}$ of ordered triples of distinct elements, the \textsc{Betweenness} problem asks whether there exists a vertex ordering $\sigma$ such that for every $(a,b,c)\in\mathcal{B}$, either $a<_{\sigma} b<_{\sigma} c$ or $c<_{\sigma} b<_{\sigma} a$. 
An ordered triple of $\mathcal{B}$ will be called a \emph{betweenness constraint}. 
Given a dense set $\mathcal{B}$ of betweenness constraints and a parameter $k\in\mathbb{N}$, the parameterized \BIT{} problem consists in computing a subset $\mathcal{F}\subseteq \mathcal{B}$ of at most $k$ betweenness constraints to edit in order to obtain a consistent set of betweenness constraints.
\end{itemize}

\paragraph{Known results.}
The parameterized \FAST{} problem is a well-studied problem from the combinatorial~\cite{EM65,RP70}  as well as from the algorithmic viewpoint~\cite{Alon06,MathieuS07}. It is known to be NP-complete~\cite{Alon06,CTY07}, but fixed parameter tractable~\cite{ALS09,KS10,RS06}. The first kernelization algorithms for \FAST{}~\cite{ALS09,DGHNT10} yield $O(k^2)$ vertex-kernels. Recently, a kernel with at most $(2 + \epsilon)k$ vertices, $\epsilon >0$, has been proposed~\cite{BFG+11}. This latter result is strongly based on a
PTAS which computes a linear vertex ordering $\sigma$ with at most $(1 + \epsilon)k$ backward arcs (\ie arcs $vu$ with $u <_\sigma v$), $\epsilon >0$. All these kernels rely on the fact that a tournament is acyclic if and only if it does not contain any directed circuit of size three.
The \FASBT{} problem is known to be $NP$-Complete~\cite{GHM07} and fixed parameter tractable~\cite{DGHNT10}. A cubic vertex-kernel was first obtained for this problem by Misra et al.~\cite{MRRS13}, which has been recently improved to a quadratic vertex-kernel by Guo and Xiao~\cite{GX13} through the concept of so-called \emph{bimodules}. Here again, a well-known result states that a bipartite tournament is acyclic if and only if it does not contain any directed circuit on four vertices. 
The parameterized \RTI{} problem is a classical problem from phylogenetics (see e.g.~\cite{BGJ10,SS03}). It is known to be NP-complete~\cite{BGJ10} but fixed parameter tractable~\cite{GB10,GM10}, the fastest algorithm running in time $O(n^ 4) + 2^ {O(k^ {1/3}logk)}$~\cite{GM10}. Moreover, Guillemot and Mnich~\cite{GM10} provided a quadratic vertex-kernel for \RTI{}. 
Finally, the parameterized \BIT{} problem is NP-Complete~\cite{AA07} but fixed parameter tractable~\cite{KS10,S09}.

\paragraph{Contributions.}
It is worth to notice that the four problems we tackle share some similar properties.
They all can be formalized as modification problems on a dense set of fixed arity constraints. For example, in the case of \textsc{Feedback Arc Set}, a constraint is  represented by an arc $uv$ to express that $u$ has to occur before $v$ in the sought vertex ordering. Likewise the rooted triplets and the betweenness constraints form arity three constraints.
They all address complete or dense instances for which the acyclicity or the consistency can be characterized by means of small bounded obstructions: directed circuit of size $3$ for tournaments [folklore] and of size $4$ for bipartite tournaments~\cite{DGHNT10}; inconsistent sub-instance of size $4$ for dense sets of rooted triplets~\cite{BD86,GM10} and for dense set of betweenness constraints (see Lemma~\ref{lem:consistentbit}).

Approaches similar to Conflict Packing have already been used in~\cite{BKM09,WNFC10}. In this paper, using Conflict Packing we are able to either improve the best known kernel bound 
or provide alternative and/or simpler proof of the best bound.
In Section~\ref{sec:fast}, we present a linear-vertex kernel for the parameterized \FAST{} problem and a quadratic one for the parameterized \FASBT{} problem. Such bounds were already known to exist~\cite{BFG+11,GX13}. However, we believe that the proof of \FAST{} is simpler, and we find interesting to propose an unified framework to cope with such problems. Together with the Conflict Packing technique, a key ingredient in our proofs is the notion of \emph{safe partition} introduced in~\cite{BFG+11}. Next, we use the Conflict Packing technique and adapt the notion of safe partition to the context of trees to obtain a linear vertex-kernel for the \RTI{} problem (Section~\ref{sec:rti}). This improves the best known bound of $O(k^2)$ vertices~\cite{GM10}. Finally, we apply the \emph{Conflict Packing} technique together with a sunflower-based reduction rule to the \BIT{} problem. This allows us to obtain a linear-vertex kernel with at most $5k$ vertices (Section~\ref{sec:bit}). Based on some of the results we present here, this latter result has recently been improved to a kernel of size at most $(2+\epsilon)k+4$ vertices~\cite{Per13}.


%
%
%
%
%
%
\section{Feedback Arc Set in tournaments and bipartite tournaments}
\label{sec:fast}


In this section, we consider asymmetric and loopless directed graphs (digraphs for short). Let $D=(V,A)$ be a digraph. An arc oriented from vertex $u\in V$ to $v\in V$ is denoted $uv$. The \emph{in-neighbourhood} of a vertex $x$ is the set $N^-(x)=\{y\in V\mid yx\in A\}$ and the \emph{out-neighbourhood} of $x$ is $N^+(x)=\{y\in V\mid xy\in A\}$. If $V'\subseteq V$, then $D[V']=(V',A')$ is the subdigraph \emph{induced} by $V'$, that is $A'=\{uv \in A\mid u\in V', v\in V'\}$. Similarly, if $A'\subseteq A$, then $D[A']=(V',A')$ denotes the \emph{digraph} where $V'=\{v\in V\mid \exists uv\in A' \mbox{ or } vw\in A'\}$. We say that $D$ is \emph{acyclic} if it does not contain any circuit and that $D$ is \emph{transitive} if for every triple of vertices $\{u, v, w\}\subseteq V$ such that $uv\in A$ and $vw\in A$, then $uw\in A$. A \emph{tournament} $T=(V,A)$ is a complete asymmetric digraph and a \emph{bipartite tournament} $T=(X\uplus Y,A)$ is a complete bipartite digraph (there is an arc between $u$ and $v$ if and only if $u\in X$ and $v\in Y$). We denote by $\mathbf{FAS}(D)$ the size of a minimum feedback arc set of $D$.

A vertex ordering $\sigma$ of $D$ is a total order on the vertex set $V$. We denote by $u<_{\sigma} v$ the fact that $u$ is smaller than $v$ in $\sigma$. It is well-known that every acyclic digraph has a \emph{topological ordering} $\sigma$ of the vertices: for every arc $uv\in A$, $u<_{\sigma} v$. An \emph{ordered digraph} is a triple $D_{\sigma}=(V,A,\sigma)$ where $D=(V,A)$ is a digraph and $\sigma$ a vertex ordering of $D$. An arc $uv\in A$ in $D_{\sigma}$ is \emph{backward} if $v<_{\sigma} u$, \emph{forward} otherwise. Let $uv$ be a backward arc, then $span(uv)=\{w\in V\mid v<_{\sigma} w<_{\sigma} u\}$. A \emph{certificate} $c(uv)$ in $D_{\sigma}$ is the vertex set of a directed path from $v$ to $u$ only composed of forward arcs. Observe that every vertex $w\in c(uv)$ distinct from $u$ and $v$ belongs to $span(uv)$ and that $c(uv)$ induces a circuit in which $uv$ is the unique backward arc. We abusively say that an arc  belongs to a certificate if its two vertices do. A subset $B \subseteq A$ of \emph{backward arcs} can be \emph{certified} whenever there exists a collection $\{c(e) : e \in B\}$ of \emph{arc-disjoint} certificates.

The following Lemma on matching and vertex cover in bipartite graphs will be used in three of our kernels.

\begin{lemma} \label{lem:matching}
Let $ B = (X \uplus Y, E)$ be a bipartite graph, $S$ be a minimum vertex cover of $B$, $S_X = X \cap S$ and $S_Y = Y \cap S$. Any subset $I \subseteq S_X$ can be matched into $Y \setminus S_Y$. 
\end{lemma}
\begin{proof}
By Hall's theorem~\cite{H35}, $I$ can be matched into $Y\setminus S_Y$ iff $|I'| \leqslant |N(I') \cap (Y \setminus S_Y)|$ for every subset $I' \subseteq I$. Hence, assume for a contradiction that there exists $I'\subseteq I$ such that $|I'| > |N(I') \cap (Y \setminus S_Y)|$. As there is no edge in $B$ between $X\setminus S_X$ and $Y\setminus S_Y$ ($S$ is a vertex cover of $B$), the set $S'=(S \setminus I') \cup (N(I') \cap (Y \setminus S_Y))$ is a vertex cover of $B$ and $|S'|<|S|$: contradicting the minimality of $S$ (see Figure~\ref{fig:vcrti}). Thereby for every subset $I'\subseteq I$, we have $|I'| \leqslant |N(I') \cap (Y \setminus S_Y)|$. This concludes the proof.  
\end{proof}

\begin{figure}[h]
\begin{center}
\centerline{\includegraphics[scale=1]{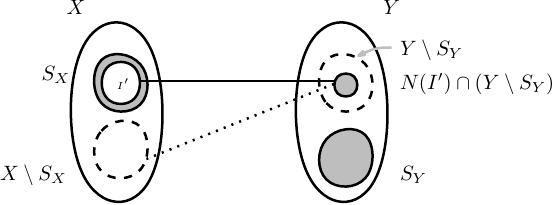}}
\caption{The case $|I'| > |N(I') \cap (Y \setminus S_Y)|$. The gray sets represent the vertex cover $S'$.\label{fig:vcrti}}
\end{center}
\end{figure}

\subsection{Reduction rules for feedback arc set in digraphs}

The following definitions and reduction rules were the core tool of the kernelization algorithm in~\cite{BFG+11}. The purpose of this section is to show how they can be easily implemented in the case of tournaments and bipartite tournaments.

A vertex of a digraph $D$ is {irrelevant} if it does not belong to any circuit of $D$. It is easy to observe that  optimal feedback arc sets of a digraph avoid  arcs incident to irrelevant vertices, which implies the following reduction rule.

\begin{polyrule} \label{rule:uselessvertexfast}
Remove every irrelevant vertex.
\end{polyrule}

A partition $\mathcal{P} = \{V_1, \dots, V_l\}$ of $V$ is an \emph{ordered partition} of an ordered digraph $D_\sigma$ if for every $i \in [l]$, $V_i$ is a set of consecutive vertices in $\sigma$. Hereafter an ordered partition of $D_\sigma$ will be denoted by $\mathcal{P}_\sigma=(V_1,\dots V_l)$, assuming that for every $x\in V_i$ and $y\in V_j$ with $1\leqslant i<j\leqslant l$, $x<_\sigma y$. The set of \emph{external arcs} of $\mathcal{P}_\sigma$ is $A_E(D_{\sigma},\mathcal{P}_\sigma)=\{uv\mid u\in V_i, v\in V_j, i\neq j\}$. A certificate $c(uv)$ of an external backward arc $uv$ is \emph{external} if the vertices of $c(uv)$ belong to distinct parts of $\mathcal{P}_\sigma$ (the arcs of the $v,u$-path are external). We say that a subset $B\subseteq A_E(D_{\sigma},\mathcal{P}_\sigma)$ of backward arcs is \emph{externally certified} if there exists a set $\{c(e)\mid e\in B\}$ of pairwise arc-disjoint external certificates.

\begin{definition}[Safe Partition] 
Let $D_\sigma = (V, A, \sigma)$ be an ordered digraph. An ordered partition $\mathcal{P}_\sigma = (V_1,\dots,V_l)$ of $D_\sigma$ is a \emph{safe partition} if the subset $B_E$ of backward arcs of $A_E(D_{\sigma},\mathcal{P}_\sigma)$ is non-empty and can be can be externally certified (see Figure~\ref{fig:spfasbt}).
\end{definition}
	
\begin{figure}[h]
\centerline{\includegraphics[scale=2]{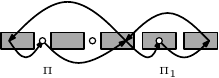}}
\caption{The two external backward arcs of the partition can be certified using the paths $\Pi$ and $\Pi_1$. In this example, there are no other external backward arcs. \label{fig:spfasbt}}
\end{figure} 
	
\begin{polyrule}[Safe partition~\cite{BFG+11}] \label{rule:safepartition}
Let $(D=(V,A),k)$ be an instance of \FAS{}, $\sigma$ be a vertex ordering of $V$ and $\mathcal{P}_\sigma = (V_1,\dots,V_l)$ be a safe partition of  $D_\sigma = (V, A, \sigma)$. Return $(D',k')$ where $D'$ is obtained from $D$ by reversing the external backward arcs $B_E\subseteq A_E(D_{\sigma},\mathcal{P}_\sigma)$ and where $k'=k-|B_E|$.
\end{polyrule}

\subsection{Tournaments}

A \emph{directed-$C_3$} (or circuit of size $3$) is a set of vertices $\{u,v,w\}$ such that $\{uv,vw,wu\}\subseteq A$. As shown by the following well-known statement, irrelevant vertices in the case of tournaments can be defined as the vertices not contained in any directed-$C_3$.

\begin{lemma}[Folklore]
\label{lem:folklore}
Let $T=(V,A)$ be a tournament. Then the following properties are equivalent: (i) $T$ is acyclic; (ii) $T$ is transitive; (iii) $T$ does not contain any directed-$C_3$.
\end{lemma}

Consequently, it will be possible to certify any backward arc $uv$ of an ordered tournament $T_{\sigma}=(V,A,\sigma)$ with a unique vertex $w\in span(uv)$ such that $\{u,v,w\}$ induces a directed-$C_3$. Clearly the cardinality of any set of arc-disjoint directed-$C_3$ provides a lower bound on the minimum feedback arc set of a tournament. We show how such a set can be used to compute in polynomial time a safe partition.

\newcommand{\T}{T_{\sigma}}
\newcommand{\V}{V_{\sigma}}

\begin{definition}[Conflict packing] \label{def:cpfast}
A \emph{conflict packing} of a tournament $T$ is a \emph{maximal} collection of arc-disjoint directed-$C_3$'s.
\end{definition}

Observe that a conflict packing can be computed in polynomial time, for instance by using a greedy algorithm. Given a conflict packing $\mathcal{C}$ of an instance $(T = (V,A),k)$ of \FAST{}, $V(\mathcal{C})$ denotes the set of vertices covered by $\mathcal{C}$. 

\begin{observation} \label{lem:cpfast}
Let $\mathcal{C}$ be a conflict packing of a tournament $T$. If $\mathbf{FAS}(T)\leqslant k$, then $|\mathcal{C}| \leqslant k$ and $|V(\mathcal{C})| \leqslant 3k$.
\end{observation}

\begin{lemma}[Conflict Packing] \label{lem:realcpfast}
Let $\mathcal{C}$ be a conflict packing of a tournament $T=(V,A)$. There exists a vertex ordering $\sigma$ of $T$ such that every backward arc $uv$ of $T_{\sigma}$ satisfies $u,v \in V(\mathcal{C})$.
\end{lemma}

\begin{proof}
Let us consider the vertex set $V'=V \setminus V(\mathcal{C})$. By maximality of $\mathcal{C}$, $T'=T[V']$ is acyclic. Observe also that for every $v\in V(\mathcal{C})$, $T[V'\cup\{v\}]$ is acyclic: otherwise $T$ would contain a directed-$C_3$ $\{v,u,w\}$ with $u,w\notin V(\mathcal{C})$ arc-disjoint from those of $\mathcal{C}$. Let $\tau$ be the unique topological ordering of $T'$. For every vertex $v\in V(\mathcal{C})$, there is a unique pair $u,w$ of  consecutive vertices in $\tau$ such that the insertion of $v$ between $u$ and $w$ yields the topological ordering of $T[V'\cup\{v\}]$. The pair $\{u,w\}$ is called the \emph{locus} of $v$. Consider any ordering $\sigma$ on $V$ such that: 
\begin{itemize}
\vspace{-0.2cm}
\item for $u,w\in V'$, if $u<_{\tau} w$, then $u\<w$, and 
\vspace{-0.2cm}
\item if $v\in V(\mathcal{C})$, and $\{u,w\}$ is the locus of $v$, then $u\< v\<w$. 
\end{itemize}
By construction of $\sigma$, every backward arc $uv$ of $T_{\sigma}$ satisfies $u,v \in V(\mathcal{C})$.
\end{proof}

\begin{theorem} \label{thm:fast}
The \FAST{} problem admits a kernel with at most $4k$ vertices.
\end{theorem}
\begin{proof}
We prove the result by showing that any instance $(T = (V,A), k)$ of \FAST{} such that $T$ is reduced under Rule~\ref{rule:uselessvertexfast} and $|V| > 4k$ is either a negative instance or can be reduced using Rule~\ref{rule:safepartition}. \\ 

Let $\mathcal{C}$ be a conflict packing of $T$ and $\sigma$ be the vertex ordering obtained through Lemma~\ref{lem:realcpfast}. 
Consider the bipartite graph $B = (I \cup V', E)$ where:
\begin{enumerate}[(i)]
\item $V' = V \setminus V(\mathcal{C})$, 
\vspace{-0.2cm}
\item there is a vertex $i_{vu}$ in $I$ for every backward arc $vu$ of $\T$ and, 
\vspace{-0.2cm}
\item $i_{vu}w \in E$ if $w \in V'$ and $\{u,v,w\}$ is a certificate of $vu$.
\end{enumerate}	 
	
Observe that any matching in $B$ of size at least $k+1$ corresponds to a conflict packing of size at least $k+1$, in which case the instance is negative (Observation~\ref{lem:cpfast}). Hence, by K\"onig-Egervary's theorem~\cite{BM76}, a minimum vertex cover $S$ of $B$ has size at most $k$. We denote $S_I=S\cap I$ and $S_{V'}=S\cap V'$.
Assume that $|V| > 4k$. Since $|S_{V'}| \leqslant k$ and  $|V(\mathcal{C})| \leqslant 3k$ (Lemma~\ref{lem:cpfast}), we have 
$V' \setminus S_{V'} \neq \emptyset$. Let $\mathcal{P}_\sigma=(V_1, \dots, V_l)$ be the ordered partition of $\T$ such that every part $V_i$ consists of either a vertex of $V' \setminus S_{V'}$ or a maximal subset of consecutive vertices (in $\sigma$) of $V\setminus (V' \setminus S_{V'})$ (see Figure~\ref{fig:spfast}).

\begin{figure}[h]
\centerline{\includegraphics[scale=1]{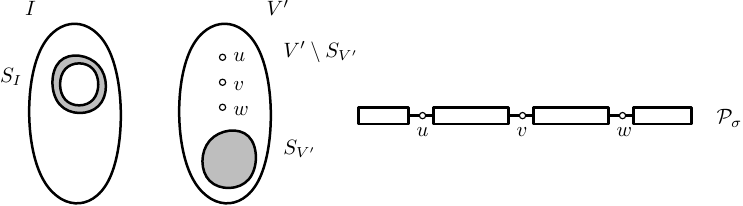}}
\caption{The construction of the ordered partition $\mathcal{P}_\sigma$ based on the bipartite graph $B$. The vertices of $V' \setminus S_{V'} = \{u,v,w\}$ are alone in their parts while the other vertices are grouped in maximal subsets of consecutive vertices (w.r.t. $\sigma)$.  \label{fig:spfast}}
\end{figure}

\begin{claim} \label{claim:pissafe}
$\mathcal{P}_\sigma$ is a safe partition of $\T$.
\end{claim}
\begin{proofclaim}
Let $w$ be a vertex of $V'\setminus S_{V'}$. By Lemma~\ref{lem:realcpfast}, $w$ is not incident to any backward arc.
Since $T$ is reduced under Rule~\ref{rule:uselessvertexfast}, there must exist a backward arc $f$ such that $w\in span(f)$. It follows that the set $A_E(T_{\sigma},\mathcal{P}_\sigma)$ of external arc contains at least one backward arc. 
Let $e=vu\in A_E(T_{\sigma},\mathcal{P}_\sigma)$ be a backward arc of $\sigma$. By construction of $\mathcal{P}_\sigma$, there exists a vertex $w_e \in (V' \setminus S_{V'})\cap span(e)$. Then $\{u,v,w_e\}$ is a certificate of $e$ and $i_{vu}w_e$ is an edge of $B$. Observe that since $S$ is a vertex cover and $w_e \notin S_{V'}$, the vertex $i_{vu}$ belongs to $S_I$ (otherwise the edge $i_{vu}w_e$ would not be covered). Thereby the subset $I'\subseteq I$ corresponding to the backward arcs of $A_E$ is included in $S_I$. 
By Lemma~\ref{lem:matching}, we know that $I'$ can be matched into $V' \setminus S_{V'}$ in $B$. 
Since every vertex of $V' \setminus S_{V'}$ is a singleton in $\mathcal{P}_\sigma$, the existence of the matching shows that the backward arcs of $A_E$ can be certified using arcs of $A_E$ only, and hence $\mathcal{P}_\sigma$ is safe. 
\end{proofclaim}

This concludes the proof. 
\end{proof}
 
\subsection{Bipartite tournaments}

A  \emph{directed-$C_4$} (or circuit of length $4$) is a subset $\{w,x,y,z\}$ of vertices such that $\{wx, xy, yz, zw\}\subseteq A$. The following statement shows that directed-$C_4$'s will play the same role as directed-$C_3$'s in tournaments. 

\begin{lemma}\cite{DGHNT10}
A bipartite tournament is acyclic if and only if it does not contain a directed-$C_4$.
\end{lemma}

\begin{lemma} \cite{MRRS13}
If a vertex of a bipartite tournament belongs to a circuit, then it belongs to a directed-$C_4$.
\end{lemma}

It follows that in the case of bipartite tournaments, an irrelevant vertex is a vertex not belonging to any directed-$C_4$. Let $T_\sigma = (X\uplus Y,A,\sigma)$ be an ordered bipartite tournament. A \emph{certificate} for a backward arc $e = wz$ of $T_\sigma$ (i.e. $z<_{\sigma} w$) is a set $c_B(e)=\{z,u,v,w\}$ such that $zu$, $uv$ and $vw$ are forward arcs. In other words, $c_B(e)$ is a directed-$C_4$ for which $e$ is the unique backward arc. 

\begin{definition} \label{def:cpfasbt}
A \emph{conflict packing} $\mathcal{C}$ of a bipartite tournament is a \emph{maximal} collection of \emph{arc-disjoint} directed-$C_4$'s.
\end{definition}

Remark that two directed-$C_4$'s of $\mathcal{C}$ may intersect on two vertices, provided that these vertices both belong to  either $X$ or $Y$. Here again, such a collection can be computed in polynomial time (using for instance a greedy algorithm). 

\begin{observation} \label{prop:cpfasbt}
Let $\mathcal{C}$ be a conflict packing of a tournament $T=(X\uplus Y, A)$. If $\mathbf{FAS}(T)\leqslant k$, then $|\mathcal{C}| \leqslant k$ and $|V(\mathcal{C})| \leqslant 4k$.
\end{observation}

The kernel in the case of bipartite tournaments requires a third reduction rule based on modules. A \emph{module} in a digraph $D=(V,A)$ is a set of vertices $M$ such that for every $u,v\in V$, $N^+(u)\setminus M=N^+(v)\setminus M$ and $N^-(u)\setminus M=N^-(v)\setminus M$. A module $M$ is \emph{non-trivial} if $|M| > 1$ and $M \neq V$. A module $M$ is \emph{maximal} if it is not included in any non-trivial module $M'$. It is easy to observe that in a bipartite tournament $T=(X\uplus Y,A)$, any non-trivial module $M$ is either a subset of $X$ or of $Y$ (and thereby an independent set). 

\begin{polyrule}[\cite{MRRS13}] \label{rule:modulesfasbt}
Let $M$ be a non-trivial module of $T$ of size at least $k + 1$. Remove all but $k + 1$ arbitrary vertices from $M$. 
\end{polyrule}

The correctness of this rule is a consequence of the following observation. Consider an instance without any irrelevant vertex. Then every vertex belongs to a $C_4$ and whenever a module $M$ of size at least $k+1$ exists, there exists a set $\{\{t,u,v,w\}\mid t\in M \}$ of at least $k+1$ directed-$C_4$'s pairwise intersecting on the arcs $uv$ and $vw$. Thereby any solution of size at most $k$ has to reverse one of those two arcs.

To compute a safe partition, we follow the main two steps described in the case of \FAST{}: 1) based on a conflict packing, compute a permutation $\sigma$ of $X\uplus Y$ such that the vertices incident to the backward arcs are covered by the conflict packing; 2) compute from $\sigma$ a safe partition of $X\uplus Y$. Let us show how to compute the vertex ordering $\sigma$.

\begin{lemma} \label{lem:module-ordering}
Let $T=(X\uplus Y,A)$ be an acyclic bipartite tournament. There exists a partial order $\tau$ on the vertices such that every topological ordering $\sigma$ of $T$ is a linear extension of $\tau$.
\end{lemma}
\begin{proof}
We define $\tau$ as follows: $x<_{\tau} y$ if and only if $xy\in A$ or there exists $z$ such that $xz\in A$ and $zy\in A$. It is easy to observe that two vertices are incomparable with respect to $\tau$ if and only if they belong to the same module. Let us prove that if $\sigma$ is a topological ordering of $T$, then the set of vertices of any maximal module $M$ of $T$ occurs consecutively in $\sigma$, implying the result.

Assume by contradiction that there exist $x,y\in M$ and $v\notin M$ such that $x<_{\sigma}v<_{\sigma} y$. W.l.o.g. assume that $M\subseteq X$. First consider that $v\in Y$. If  $v\in N^+(x)$ (the case  $v\in N^-(x)$ is similar) then $v\in N^+(y)$, in other words the arc $yv$ is backward in $\sigma$: contradicting the fact that $\sigma$ is a topological ordering. Assume that $v\in X$. As $v\notin M$, there exists a vertex $s\notin M$ distinguishing $v$ from $x,y$: either $s\in N^+(v)$, $s\in N^-(x)$ and $s\in N^-(y)$ or $s\in N^-(v)$, $s\in N^+(x)$ and $s\in N^+(y)$. In any case, one of the three arcs between $s$ and $x, u$ or $y$ is backward in $\sigma$: contradicting the fact that $\sigma$ is a topological ordering.
%
\end{proof}

\begin{lemma} \label{lem:cpfasbt}
Let $\mathcal{C}$ be a conflict packing of a bipartite tournament $T = (X\uplus Y,A)$. There exists a vertex ordering $\sigma$ of $T$ such that every backward arc $uv$ of $T_\sigma$ (i.e. $v<_{\sigma} u$) satisfies $u,v \in V(\mathcal{C})$. 
\end{lemma}
\begin{proof}
Let $V'$ be the set of vertices not covered by $\mathcal{C}$, i.e. $V'=(X\uplus Y)\setminus V(\mathcal{C})$. By maximality of $\mathcal{C}$, $T'=T[V']$ is acyclic. Let $\tau$ be the partial order on $V'$ defined by Lemma~\ref{lem:module-ordering}. We define a vertex ordering $\sigma$ by extending $\tau$ (adding comparabilities) and inserting vertices of $V(\mathcal{C})$.
Let $x,y$ be two arbitrary vertices of $T$:

\begin{enumerate}
\item $x\in V(\mathcal{C})$ and $y\in V(\mathcal{C})$:
\begin{enumerate}
\item $x\in X$ and $y\in X$: if $N^-(x)\cap V'=N^-(y)\cap V'$, then $x$ and $y$ are incomparable, otherwise if $N^-(x)\subset N^-(y)$, then $x<_{\sigma} y$; 
\item  $x\in X$ and $y\in Y$: if  $\exists v\in N^-(y)\cap V'$ such that $\forall u\in N^-(x)\cap V'$, $u<_{\tau} v$, then $x<_{\sigma} y$, otherwise $y<_{\sigma} x$;
\end{enumerate}

\item If $x\in V'$ and $y\in V(\mathcal{C})$: 
\begin{enumerate}
\item $x\in X$ and $y\in X$: 
if $\exists u\in N^-(y)\cap V'$ such that $x<_{\tau} u$, then $x<_{\sigma} y$, otherwise $y<_{\sigma} x$;
\item  $x\in X$ and $y\in Y$: if $x\in N^-(y)\cap V'$, then $x<_{\sigma} y$, otherwise $y<_{\sigma} x$;
\end{enumerate}

\item If $x\in V'$ and $y\in V'$: if $x$ and $y$ are not comparable in $\tau$ and $\exists u\in V(\mathcal{C})$ with $x\in N^-(u)$ and $y\in N^+(u)$, then $x<_{\sigma} y$;
\end{enumerate}

The following claims will help to prove that $\sigma$ is a partial order.
 
\begin{claim} \label{cl:comp}
Let $x$ and $y$ be two vertices of $V(\mathcal{C})\cap X$ (resp. $V(\mathcal{C})\cap Y$), then $(N^+(x)\cap V')\subseteq (N^+(y)\cap V')$ or vice versa.
\end{claim}
\begin{proofclaim}
Assume there exist $u\in (N^+(x)\cap V')\setminus (N^+(y)\cap V')$ and $v\in (N^+(x)\cap V')\setminus (N^+(y)\cap V')$. Then $\{x, u, y, v\}$ is a directed-$C_4$ which is arc-disjoint from every directed-$C_4$ of $\mathcal{C}$, contradicting the maximality of $\mathcal{C}$.
\end{proofclaim}

\begin{claim} \label{cl:consec}
Let $y$ be a vertex of $V(\mathcal{C})$. If $x\in V'$ and $z\in V'$ are two comparable vertices in $\tau$ such that $x\in N^-(y)$ and $z\in N^+(y)$, then $x<_{\tau} z$.
\end{claim}
\begin{proofclaim}
Assume that $z<_{\tau} x$. As $x$ and $z$ do not belong to the same module of $T'$ there exists $u\in V'$ with $z<_{\tau} u<_{\tau} x$ such that $zu$ and $ux$ are arcs of $T$. It follows that the vertices $zuxy$ form a directed-$C_4$ arc-disjoint from those of $\mathcal{C}$: contradicting the maximality of $\mathcal{C}$.
\end{proofclaim}

\begin{claim} \label{cl:transitive}
Let $x,z$ be two vertices such that $x\in V'\cap X$ and $z\in V(\mathcal{C})\cap Y$ (resp. $x\in V'\cap Y$ and $z\in V(\mathcal{C})\cap X$). If there exists $y\in V'$ such that $x<_{\tau} y$ and $y<_{\sigma} z$, then $x\in N^-(z)$ and $x<_{\sigma} z$.
\end{claim}
\begin{proofclaim}
We consider the case where $x \in V' \cap X$ and $z \in V(\mathcal{C}) \cap Y$, the other case being symmetric.
By condition 2 of definition of $\sigma$, $y<_{\sigma} z$ implies that if $y\in X$, then $y\in N^-(z)$ otherwise there exists $u\in N^-(z)\cap V'$ such that $y<_{\tau} u$. Assume that $y\in X$. As $x<_{\tau} y$, $x$ and $y$ do not belong to the same maximal module. So there exists $v\in Y\cap V'$ such that $x<_{\tau} v<_{\tau} y$. Observe that if $x\in N^+(z)$, then the vertices $\{x,v,y,z\}$ form a directed-$C_4$. Likewise if $y\in Y$, the vertices $\{x,y,u,z\}$ form a directed-$C_4$. In both cases we obtain a directed-$C_4$ arc-disjoint from those of $\mathcal{C}$: contradicting the maximality of $\mathcal{C}$. It follows that $x\in N^-(z)$ and thereby $x<_{\sigma} z$ (condition 2b of definition of $\sigma$).
\end{proofclaim}

\begin{claim} \label{cl:PO}
$\sigma$ is a partial order. 
\end{claim}
\begin{proofclaim}
Let us first observe that $\sigma$ is antisymmetric. If $x$ or $y$ belong to $V(\mathcal{C})$, antisymmetry is straightforward from definition of $\sigma$. Suppose that $x$ and $y$ both belong to $V'$. As new comparabilities are added only between pair of vertices non-comparable in $\tau$, antisymmetry follows from Claim~\ref{cl:comp}. 

Let us prove that $\sigma$ is transitive with a case by case analysis.  Keep in mind that two vertices are incomparable with respect to $\tau$ if and only if they belong to the same maximal module of $T$. Let $x, y, z$ be three vertices such that $x<_{\sigma} y$ and $y<_{\sigma} z$.
\begin{itemize}
\item $x,y,z\in V'$: If $x<_{\tau} y$ and $y<_{\tau} z$, as $\tau$ is transitive, we also have $x<_{\tau} z$ and thereby $x<_{\sigma} z$. Suppose that $x, y$ belong to the same maximal module of $T'$ and $y<_{\tau} z$, then by construction of $\tau$, we also have $x<_{\tau} z$ implying $x<_{\sigma} z$. The same argument holds if  $x<_{\tau} y$ and $y,z$ belong to the same maximal module of $T'$. The last case to consider is when $x,y,z$ belong to the same maximal module of $T'$. Then $x<_{\sigma} y$ implies the existence of $u\in V(\mathcal{C})$ such that $x\in N^-(u)$ and $y\in N^+(u)$ (condition 3 of definition of $\sigma$). Likewise, $y<_{\sigma} z$ implies the existence of $v\in V(\mathcal{C})$ such that $y\in N^-(v)$ and $z\in N^+(v)$. As $y\in N^-(v)\setminus N^-(u)$, Claim~\ref{cl:comp} implies that $N^-(u)\cap V'\subset N^-(v)\cap V'$ and thereby $x\in N^-(v)$. As $z\in N^+(v)$, it follows that $x<_{\sigma} z$ (condition 3 of definition of $\sigma$).

\item $x,y \in V'$ and $z\in V(\mathcal{C})$ (or $x\in V(\mathcal{C})$ and $y,z\in V'$): Assume that $x$ and $y$ belong to the same maximal module of $T'$. 
It follows that $x$ and $y$ both belong either to $X$ or to $Y$. Assume w.l.o.g. that $x,y\in X$.
If $z\in X$, then $y<_{\sigma} z$ implies the existence of $u\in N^-(z)\cap V'$ such that $y<_{\tau} u$ (condition 2a of definition of $\sigma$). As $x$ and $y$ belongs to the same module of $T'$, $y<_{\tau} u$ implies that $x<_{\tau} u$ and thereby $x<_{\sigma} z$. Assume that $z\in Y$. Then $y<_{\sigma} z$ implies $y\in N^-(z)$ (condition 2b of definition of $\sigma$). Moreover, $x <_\sigma y$ implies the existence of $u \in V(\mathcal{C})$ such that $x \in N^-(u)$ and $y \in N^+(u)$ (condition 3 of definition of $\sigma$). Observe that as $y\in N^+(u)\cap N^-(z)$, $u\neq z$. It follows that $x\in N^-(z)$ as well, 
since otherwise $\{x,u,y,z\}$ would form a directed-$C_4$ arc-disjoint from those of $\mathcal{C}$, contradicting the maximality of $\mathcal{C}$. Thereby we have $x<_{\sigma} z$.

If $x\in X$ and $z\in Y$, as $x<_{\tau} y$ and $y<_{\sigma} z$, by Claim~\ref{cl:transitive} we have $x<_{\sigma} z$. Suppose that $x,z\in X$. If $y\in X$, then $y<_{\sigma} z$ implies the existence of $u\in N^-(z)\cap V'$ such that $y<_{\tau} u$ (condition 2a). It follows that $x<_{\tau} u$. If $y\in Y$, then $y<_{\sigma} z$ implies that $y\in N^-(z)\cap V'$ (condition 2b). In both cases, there exists a vertex $v$ such that $v\in N^-(z)\cap V'$ and $x<_{\tau} v$. It follows that $x<_{\sigma} z$ (condition 2a). The other cases are symmetric. 

\item $x,z\in V'$ and $y\in V(\mathcal{C})$: 
We first consider the case where $x$ and $z$ do not belong to the same module of $T'$, implying that either $x<_\tau z$ or $z<_\tau x$. We prove that $z<_\tau x$ would contradict the maximality of $\mathcal{C}$. Suppose that $x\in X$, $z\in Y$ and that $y\in X$ (the case $y\in Y$ is symmetric). As $x<_\sigma y$, there exists $u\in N^-(y)\cap V'$ such that $x<_\tau u$ (condition 2a), or equivalently $x\in N^-(u)$. It follows that if $z<_\tau x$ or equivalently $z\in N^-(x)$, we obtain a directed-$C_4$ on $\{x,y,z,u\}$ arc-disjoint from those of $\mathcal{C}$. Suppose that $x\in X$ and $z\in X$. If $z<_\tau x$, then there exists $u\in Y\cap V'$ such that $u\in N^+(z)\cap N^-(x)$. If $y\in Y$, then $x<_\sigma y$ and $y<_{\sigma} z$ implies $y\in N^-(z)\cap N^+(x)$ (condition 2b). So $\{x,y,z,u\}$ forms a directed-$C_4$ arc-disjoint from those of $\mathcal{C}$. If $y\in X$, then there exists $v\in N^-(y)\cap V'$ such that $x<_\tau v$ (condition 2a). Observe that $u\in N^-(y)$ would contradict $y<_\sigma z$ (condition 2b). Thereby $\{x,y,u,v\}$ forms a directed-$C_4$ arc-disjoint from those of $\mathcal{C}$. So as announced, we prove that if $x$ and $z$ do not belong to the same module,  $x<_\tau z$.

Let us now assume that $x, z$ belong to the same maximal module $M$ of $T'$. So w.l.o.g. assume that $x,z\in X$. Suppose that $y\in X$. Then $x<_{\sigma} y$ implies the existence of $u\in N^-(y)\cap V$ with $x<_{\tau} u$ (condition 2a). Observe that $xu$ is an arc of $T$ ($u\in Y$ and $x<_{\tau} u$) and thereby $zu$ is an arc of $T$ as well ($x,z$ belong to the same maximal module not including $u$). It follows that $z<_{\sigma} y$ (condition 2a): contradiction. So we have $y\in Y$. Observe then that $x<_{\sigma} y$ and $y<_{\sigma} z$ implies that $x\in N^-(y)$ and $z\in N^+(y)$. Hence we obtain $x<_{\sigma} z$ (condition 3).

\item $x,y \in V(\mathcal{C})$ and $z\in V'$: If $x\in X$ and $z\in Y$, we prove that $z\in N^+(x)$ implying that $x<_{\sigma} z$ (condition 2b). Assume that $y\in X$. As $x<_{\sigma} y$ implies $N^-(x)\cap V' \subset N^-(y)\cap V'$ (condition 1a) and as $y<_{\sigma} z$ implies $z\in N^+(y)$ (condition 2b), we have $z\in N^+(x)$. Assume that $y\in Y$. Then $x<_{\sigma} y$ implies the existence of $v\in N^-(y)\cap V'$ such that for every $u\in N^-(x)\cap V'$, $u<_{\tau} v$ (condition 1b). Moreover $y<_{\sigma} z$ implies that for every $w\in N^-(y)\cap V'$, $w<_{\tau} z$ (condition 2a). It follows that $v<_{\tau} z$. To conclude, observe now that $z\in N^-(x)$ would contradict $v<_\tau z$ (condition 1b). 

If $x,z\in Y$, we prove that every $u\in N^-(x)\cap V'$ satisfies $u<_{\tau} z$, implying that $x<_{\sigma} z$ (condition 2a). Assume that $y\in X$. Then $x <_{\sigma} y$ implies the existence of $v\in N^-(y)\cap V'$ such that for every $u\in N^-(x)\cap V'$, $u<_{\tau} v$. Moreover $y<_{\sigma} z$ implies that $y\in N^-(z)$ (condition 2b). By Claim~\ref{cl:consec}, as $v\in N^-(y)$ and $z\in N^+(y)$, we have $v<_{\tau} z$. Thereby every $u\in N^-(x)\cap V'$ satisfies $u<_{\tau} z$. Assume that $y\in Y$. Since $x<_{\sigma} y$ implies that $N^-(x)\cap V'\subset N^-(y)\cap V'$ (condition 1a) and since $y<_{\sigma} z$ implies that every $v\in N^-(y)\cap V'$ satisfies $v<_{\tau} z$ (condition 2a), we have that every $u\in N^-(x)\cap V'$ satisfies $u<_{\tau} z$.

\item $x,z\in V(\mathcal{C})$ and $y\in V'$: If  $x\in X$ and $z\in Y$, we prove the existence of a vertex $v\in N^-(z)\cap V'$ such that for every $u\in N^-(x)\cap V'$, $u<_{\tau} v$, implying that $x<_{\sigma} z$ (condition 1b). This is the case if $y\in X$, as $x<_{\sigma} y$ implies that every $u\in N^-(x)\cap V'$ satisfies $u<_{\tau} y$ (condition 2a) and as $y<_{\sigma} z$ implies $y\in N^-(z)$ (condition 2b). So assume that $y\in Y$. Then $x<_{\sigma} y$ implies $y\in N^+(x)$ (condition 2b). Moreover $y<_{\sigma} z$ implies the existence of $v\in N^-(z)\cap V'$ such that $y<_{\tau} v$ (condition 2a). It follows by Claim~\ref{cl:consec} that every $u\in N^-(x)\cap V$ satisfies $u<_{\tau} v$.

If  $x,z\in X$, we prove that $N^-(x)\cap V'\subset N^-(z)\cap V'$, implying that $x<_{\sigma} z$ (condition 1a). Assume that $y\in X$. Then  $x<_{\sigma} y$ implies that every $u\in N^-(x)\cap V'$ satisfies $u<_{\tau} y$ (condition 2a). Moreover $y<_{\sigma} z$ implies the existence of $v\in N^-(z)\cap V'$ such that $y<_{\tau} v$ (condition 2a). It follows that $v\in (N^-(z)\setminus N^-(x))\cap V'$. Assume that $y\in Y$. As $x<_{\sigma} y$ implies $y\in N^+(x)\cap V'$ (condition 2b) and as $y<_{\tau} z$ implies $y\in N^-(z)\cap V'$ (condition 2b), we have $y\in N^-(z)\setminus N^-(x))\cap V'$. In both cases, Claim~\ref{cl:comp} implies that $N^-(x)\cap V'\subset N^-(z)\cap V'$.

\item $x\in V'$ and $y,z\in V(\mathcal{C})$: If  $x\in X$ and $z\in Y$, we prove that $x\in N^-(z)$ implying that $x<_{\sigma} z$ (condition 2b). Assume that $y\in X$. Observe that $x<_{\sigma} y$ implies the existence of $u\in N^-(y)\cap V'$ such that $x<_{\tau} u$ (condition 2a). Moreover $y<_{\sigma} z$ implies the existence of $v\in N^-(z)\cap V'$ such that every $w\in N^-(y)\cap V$ satisfies $w<_{\tau} v$ (condition 1b). It follows that $x<_{\tau} v$ and by Claim~\ref{cl:transitive}, $x\in N^-(z)$. Assume that $y\in Y$. As $x<_{\sigma} y$ implies $x\in N^-(y)$ (condition 2b) and as $y<_{\sigma} z$ implies $N^-(y)\cap V'\subset N^-(z)\cap V'$ (condition 1a), we have that $x\in N^-(z)$.

If $x,z\in X$, we prove the existence of $u\in N^-(z)\cap V'$ such that $x<_{\tau} u$, implying that $x<_{\sigma} z$ (condition 2a). Assume that $y\in X$. Observe that $x<_{\sigma} y$ implies the existence of $u\in N^-(y)\cap V'$ such that $x<_{\tau} u$ (condition 2a). Moreover $y<_{\sigma} z$ implies $N^-(y)\cap V'\subset N^-(z)\cap V'$ (condition 1a). It follows that $u\in N^-(z)\cap V'$ and $x<_{\tau} u$. Assume that $y\in Y$. Then $x<_{\sigma} y$ implies $x\in N^-(y)$ (condition 2b). Moreover $y<_{\sigma} z$ implies the existence of $u\in N^-(z)$ such that every $v\in N^-(y)\cap V'$ satisfies $v<_{\tau} u$ (condition 1a). It follows that  $x<_{\tau} u$.

\item $x,y,z\in V(\mathcal{C})$: If  $x\in X$ and $z\in Y$, we prove the existence of $v\in N^-(z)\cap V'$ such that every $u\in N^-(x)\cap V'$ satisfies $u<_{\tau} v$, implying that $x<_{\sigma} z$ (condition 1b). Assume that $y\in X$. Then $x<_{\sigma} y$ implies $N^-(x)\cap V'\subset N^-(y)\cap V'$ (condition 1a). Moreover $y<_{\sigma} z$ implies the existence of $v\in N^-(z)\cap V'$ such that every $w\in N^-(y)\cap V'$ satisfies $w<_{\tau} v$ (condition 1b). It follows that for every $u\in N^-(x)\cap V'$ we have $u<_{\tau} v$. Assume that $y\in Y$. Then $x<_{\sigma} y$ implies the existence of $v\in N^-(y)\cap V'$ such that every $u\in N^-(x)\cap V'$ satisfies $u<_{\tau} v$ (condition 1b). To conclude it suffices to observe that as $y<_{\sigma} z$ implies $N^-(y)\cap V'\subset N^-(z)\cap V'$ (condition 1a), $v\in N^-(z)\cap V'$.

If $x,z\in X$, we prove that $N^-(x)\cap V'\subset N^-(z)\cap V'$ implying that $x<_{\sigma} z$ (condition 1a). Assume that $y\in X$. As $x<_{\sigma} y$ implies $N^-(x)\cap V'\subset N^-(y)\cap V'$ (condition 1a) and  as $y<_{\sigma} z$ implies  $N^-(y)\cap V'\subset N^-(z)\cap V'$ (condition 1a), we have that $N^-(x)\cap V'\subset N^-(z)\cap V'$. Assume that $y\in Y$. Then $x<_{\sigma} y$ implies the existence of $v\in N^-(y)\cap V'$ such that every $u\in N^-(x)\cap V'$ satisfies $u<_{\tau} v$ (condition 1b). Moreover $y<_{\sigma} z$ implies the existence of $w\in N^-(z)\cap V'$ such that every $v\in N^-(y)\cap V'$ satisfies $v<_{\tau} w$ (condition 1b). It follows that $w\in N^-(z)\setminus N^-(x)$ and thereby Claim~\ref{cl:comp} shows that $N^-(x)\cap V'\subset N^-(z)\cap V'$.
\end{itemize} \end{proofclaim}

\begin{claim} 
If $xy$ is an arc of $T$ such that $y<_{\sigma} x$, then $x$ and $y$ both belong to $V(\mathcal{C})$.
\end{claim}
\begin{proofclaim}
Let $xy$ be an arc of $T$. By construction, $\sigma$ restricted to $V'$ is an extension of $\tau$. 
It follows that if $x$ and $y$ are two vertices of $V'$, then $x<_{\sigma} y$. If $x\in V'$ and $y\in V(\mathcal{C})$, then by construction of $\sigma$ (condition 2b of definition of $\sigma$), $x\in N^-(y)$ implies $x<_{\sigma} y$. The case $x\in V(\mathcal{C})$ and $y\in V'$ is symmetric.
\end{proofclaim}

It follows from the above claims that any linear extension of the partial order $\sigma$ (Claim~\ref{cl:PO}) is a vertex ordering satisfying the statement.
\end{proof} 	


\begin{theorem} \label{thm:kernelfasbt}
The \FASBT{} problem admits a kernel with $O(k^2)$ vertices. 
\end{theorem}

\begin{proof}

We prove the result by showing that any instance $(T = (X\uplus Y,A), k)$ of \FASBT{} such that $T$ is reduced under Rule~\ref{rule:uselessvertexfast} and Rule~\ref{rule:modulesfasbt} and $|X\uplus Y| \geqslant 11k^2+16k+1$ is either a negative instance or can be reduced using Rule~\ref{rule:safepartition}. \\ 

Let $\mathcal{C}$ be a conflict packing of $T$ and $\sigma$ be an ordering of $X\uplus Y$ satisfying the conditions of Lemma~\ref{lem:cpfasbt} (notice that each can be computed in polynomial time). 
We will denote by $V'$ the set $(X\cup Y) \setminus V(\mathcal{C})$. A directed-$P_3$ is a triple of vertices $\{x,y,z\}$ such that $xy\in A$ and $yz\in A$. A directed $P_3$ is \emph{free} if $\{x,y,z\}\subseteq V'$ and it does not exist $u\in V(\mathcal{C})$ such that $x<_{\sigma} u<_{\sigma} z$ (observe that by construction of $\sigma$ we have $x<_{\sigma} y<_{\sigma} z$). If $\{x,y,z\}$ is a free directed-$P_3$, then there exist $u,v\in V(\mathcal{C})$ such that $uv$ is a backward arc of $T_{\sigma}$ and $\{x,y,z\}\subseteq span(uv)$, that is $v<_{\sigma} x<_{\sigma} y<_{\sigma} z<_{\sigma} u$ (otherwise, $T$ would not be reduced under Rule~\ref{rule:uselessvertexfast}). Observe that independent of whether $u\in X, v\in Y$ or $u\in Y$, $v\in X$, the subset of vertices $\{v,x,y,z,u\}$ contains a certificate of $uv$. Indeed suppose without loss of generality that $u\in X$ and $v\in Y$.  If $x\in X$, then $\{u,v,x,y\}$ is a directed-$C_4$, otherwise $\{u,v,y,z\}$ is a directed-$C_4$.

We construct a bipartite graph $B=(I\cup P,E)$ as follows: 1) $I$ contains a vertex $i_{uv}$ for every backward arc $uv$ of $T_{\sigma}$; 2) $P$ contains a vertex $p_{xyz}$ for every free directed-$P_3$ of $T_{\sigma}$; and 3) $(i_{uv},p_{xyz})\in E$ if and only if $\{x,y,z\}\subseteq span(uv)$. Clearly $B$ can be built in polynomial time. It is easy to observe that any matching in $B$ witnesses an arc-disjoint set of directed-$C_4$'s. It follows that if the size $\mu$ of a maximum matching in $B$ is $k+1$ or more, then $T$ is a negative instance. Let us now explain how to compute a safe partition when $\mu\leqslant k$. We first show that as $|X\uplus Y| \geqslant 11k^2+16k+1$, we can extract from $T_{\sigma}$ a set $\mathcal{Q}$ of at least $k+1$ vertex disjoint free directed-$P_3$'s. Let us call \emph{interval of $\sigma$} a maximal subset $I\subseteq V'$
of vertices appearing consecutively in $\sigma$. 

\begin{claim} \label{cl:module-interval}
If $M$ is a module of $T$, then $M\setminus V(\mathcal{C})$ is contained in a unique interval of $V'$.
\end{claim}
\begin{proofclaim}
First observe that $M'=M\setminus V(\mathcal{C})$ is a module of $T[V']$. To compute $\sigma$, we started from a topological ordering $\tau$ of $T[V']$ and by Lemma~\ref{lem:module-ordering}, vertices of $M'$ are consecutive in $\tau$.
Assume that vertices of $M'$ are not consecutive in $\sigma$. Then by construction of $\sigma$, $V(\mathcal{C})$ contains a vertex $u$ such that $N^-(u)\cap M'\neq\emptyset$ and $N^+(u)\cap M'\neq\emptyset$. Now to be a module of $T$, $M$ has to contain vertex $u$: contradicting the fact that every module of a bipartite tournament is an independent set.
\end{proofclaim}

As by Observation~\ref{prop:cpfasbt}, $|V(\mathcal{C})|\leqslant 4k$, the set $V'$ has size at least $11k^2+12k+1$ and is partitioned in at most $4k-1$ intervals. Let $m$ be the number of maximal modules of $T$ intersecting $V'$. As $T$ is reduced under Rule~\ref{rule:modulesfasbt}, every maximal module has size at most $k+1$ and thus $m\geqslant 11k+1$. By Claim~\ref{cl:module-interval}, $m=\sum_{I\in\mathcal{I}} m_I$ where $\mathcal{I}$ is the set of intervals and $m_I$ is the number of maximal modules of $T$ intersecting $I$. For every interval $I$, define $r_I=m_I\mod 3$ and $k_I=(m_I-r_I)/3$. Observe that $\sum_{I\in\mathcal{I}} r_I\leqslant 8k+2$ implying that $\sum_{I\in\mathcal{I}} k_I\geqslant k+1$. In other words, it is possible to extract from the maximal modules of $T$ intersecting $V'$ a set $\mathcal{Q}$ of at least $k+1$ vertex disjoint free directed-$P_3$'s, each directed-$P_3$ being composed by three vertices from distinct consecutive modules.

Recall that the bipartite graph has a maximum matching of size $\mu\leqslant k$ and that $|P|=|\mathcal{Q}|\geqslant k+1$. By K\"onig-Egervary's theorem~\cite{BM76}, $B$ has a vertex cover $S$ of size $\mu$. Let $S_I=S\cap I$ and $S_P=S\cap P$. It follows that $P\setminus S_P\neq\emptyset$. Let $\mathcal{P}_\sigma$ be the ordered partition defined as follows: for every vertex $p_{xyz}\in  P\setminus S_P$, $\{x\}$, $\{y\}$ and $\{z\}$ form a singleton parts of $\mathcal{P}_\sigma$; the other parts are the maximal set of consecutive vertices not covered by the directed-$P_3$'s of $S_P\setminus P$. 

\begin{claim} \label{cl:safe-partition}
The ordered partition $\mathcal{P}_\sigma$ is a safe partition.
\end{claim}
\begin{proofclaim}
Let us first argue that the set $A_E(T_\sigma,\mathcal{P}_\sigma)$ of external arcs contains at least one backward arc. As $P\setminus S_P\neq\emptyset$, the ordered partition $\mathcal{P}_\sigma$ is non-trivial. Let $\{x,y,z\}$ be a directed-$P_3$ such that $p_{xyz}\in P\setminus S_P$. By construction, $x$, $y$ and $z$ are not in $V(\mathcal{C})$ and thereby not incident to any backward arc. As $T$ is reduced by Rule~\ref{rule:uselessvertexfast}, $z$ belongs to a directed-$C_4$ and since the directed-$P_3$ $\{x,y,z\}$ is free, there exists a backward arc $uv$ such that $v<_{\sigma} x<_{\sigma} y<_{\sigma} z<_{\sigma} u$: proving that  $A_E(T_\sigma,\mathcal{P}_\sigma)$ contains at least one backward arc.

Let us now consider $uv$ an external backward arc of $\mathcal{P}_\sigma$. By construction of $\mathcal{P}_\sigma$, there exists $p_{xyz}\in P\setminus S_P$ such that $x,y,z\in span(uv)$. Thereby $(i_{uv},p_{xyz})$ is an edge of $B$ and $\{u,x,y,z,v\}$ contains a certificate $c_I(uv)$ (either $\{u,v,x,y\}$ or $\{u,v,y,z\}$ induces a directed-$C_4$). Thereby as $S$ is a vertex cover of $B$ and $p_{xyz}\in P\setminus S_P$, $i_{uv}$ belongs to $S$. It follows that the subset $I_E\subseteq I$ corresponding to the external backward arcs is included in $S_B$. By Lemma~\ref{lem:matching}, there is a matching between $I_E$ and $P\setminus S_P$. This shows that the set external backward arcs can be certified using external arcs only, proving that $\mathcal{P}_\sigma$ is safe.
\end{proofclaim}

This concludes the proof.
\end{proof}

\section{Linear vertex-kernel for \sc{\RTI{}}}
\label{sec:rti}


Let $T$ be a rooted binary tree. The set of leaves of $T$ is denoted $V(T)$ (or $V$ when the context is clear) and we say that $T$ is \emph{defined over $V$}. Hereafter the term \emph{node} stands for internal nodes of $T$.
If $S\subseteq V$, then $T_{|S}$ denotes the rooted binary tree defined over $S$ which is homeomorphic to the subtree of $T$ spanning the leaves of $S$. If $x$ is a node of $T$, then $T_x$ denotes the subtree of $T$ rooted in $x$. 

A \emph{rooted triplet} $t$ is a rooted binary tree defined over a set of three leaves $V(t)=\{a,b,c\}$.  We denote by $t=ab|c$ the rooted triplet such that $a$ and $b$ are siblings of a child of the root of $t$, the other child of the root being $c$. We also say that $t$ \emph{chooses} $c$. A system of rooted triplets is a pair $R= (V, \mathcal{R})$ such that $\mathcal{R}$ is a set of rooted triplets each of which is defined over a triplet $\{a,b,c\}\subseteq V$. 
We say that $R= (V, \mathcal{R})$ is \emph{dense} if  $\mathcal{R}$ contains a rooted triplet for every triplet of distinct leaves $\{a,b,c\}\subseteq V$. If $S \subseteq V$, then the \emph{induced system of rooted triplets} $R[S]$ is the pair $(S,\mathcal{R}[S])$, where $\mathcal{R}[S]=\{t\in \mathcal{R}\mid V(t)\subseteq S\}$.

A rooted triplet $t$ over $\{a,b,c\}$ is \emph{consistent with} a rooted binary tree $T$ if $T_{|\{a,b,c\}}$ is homeomorphic to $t$, and \emph{inconsistent with} $T$ otherwise. We say that $R = (V,\mathcal{R})$ is \emph{consistent} if there exists a rooted binary tree $T$ defined over $V$ such that every $t \in \mathcal{R}$ is consistent with $T$. If such a tree does not exist, then $R$ is \emph{inconsistent}. Observe that if $R = (V,\mathcal{R})$ is dense and consistent, then there exists a unique binary tree $T$ with which every rooted triplet of $\mathcal{R}$ is consistent. 
Let $t$ and $t$' be two distinct rooted triplets over the same set of leaves, then \emph{editing $t$ into $t'$} means replacing $t$ by $t'$ in $\mathcal{R}$. If $\mathcal{S}$ is a set of rooted triplets over $V$, then \emph{editing $R$ according to $\mathcal{S}$} consists  for every rooted triplet $t'\in\mathcal{S}$, in editing $t\in \mathcal{R}$ (with $V(t)=V(t')$) into $t'$. The resulting system of rooted triplets is denoted by $R\odot \mathcal{S}$. We denote by $\mathbf{RTI}(R)$ the minimum size of a set $\mathcal{S}$ of rooted triplets such that $R\odot\mathcal{S}$ is a consistent system of rooted triplets.

An \emph{embedded system of rooted triplets (over $V$)} is a triple $R_T=(V,\mathcal{R},T)$ such that $(V,\mathcal{R})$ is dense and $T$ is a rooted binary tree defined over $V$. Given a subset $S \subseteq V$, we define $R_T[S] = (S, \mathcal{R}[S], T_{|S})$. 
\emph{Editing a subset $\mathcal{S}\subseteq \mathcal{R}$ according to the tree $T$} means editing every $t\in\mathcal{S}$ into $T_{|V(t)}$. This latter operation is denoted by $R\odot_{T}\mathcal{S}$.
By editing $R_T$, we always mean according to $T$. 
The (in)consistency of an embedded system of rooted triplets $R_T$ is defined with respect to the tree $T$. Let $t\in\mathcal{R}$ be a rooted triplet over $\{a,b,c\}\subset V$. 
We define $span(\{a,b,c\})$ (or $span(t)$) as the subset of $V(T)$ containing the leaves of $T_{lca(\{a,b,c\})}$, where $lca$ stands for \emph{least common ancestor} (see Figure~\ref{fig:span}). 
If $t$ is inconsistent with $T$, a \emph{certificate} of $t$ is a set $c(t)=\{a,b,c,d\}\subseteq V$ such that $d\in span(t)$ is a leaf not belonging to any rooted triplet of $\mathcal{R}$ inconsistent with $T$. 
If $t$ and $t'$ are two distinct rooted triplets inconsistent with $T$, the certificate $c(t)$ and $c(t')$ are \emph{triplet-disjoint} if $|c(t)\cap c(t')|\leqslant 2$. A \emph{certificate} for a set $\mathcal{F} \subseteq \mathcal{R}$  of rooted triplets inconsistent with $T$ is a set $c(\mathcal{F}) = \{c(t) : t \in \mathcal{F}\}$ of \emph{triplet-disjoint} certificates.

\begin{figure}[h]
\centerline{\includegraphics[scale=1.5]{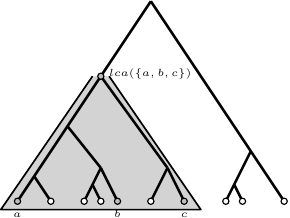}}
\caption{The notion of span for an embedded instance $R_T = (V, \mathcal{R}, T)$ of \RTI{}. The leaves of the gray subtree represents $span(\{a,b,c\})$. \label{fig:span}}
\end{figure}

\subsection{Reduction rules}

Let $R = (V, \mathcal{R})$ be a system of rooted triplets.
We say that $C\subseteq V$ (or $R[C]$ or $\mathcal{R}[C]$ depending on the context) is a \emph{conflict} if it is an inconsistent system of rooted triplets. As shown by the following lemma, conflicts will play an important role.  

\begin{lemma}[\cite{GM10}] \label{lem:localpropertyrti}
A system $R = (V, \mathcal{R})$ of rooted triplets is consistent if and only if it does not contain any conflict on $4$ leaves.
\end{lemma}

A leaf is \emph{irrelevant} if it does not belong to any conflict. The above lemma leads to the first reduction rule which can be applied in polynomial time since it is enough to look for conflicts of size $4$. A reduction rule is \emph{sound} if whenever it is applied to an instance $(R=(V,\mathcal{R}),k)$ of \RTI{} it outputs an equivalent instance $(R'=(V',\mathcal{R}'),k')$.

\begin{polyrule}[\cite{GM10}] \label{rule:rtiuseless}
Remove every irrelevant leaf $v\in V$.
\end{polyrule}

 The following lemma shows that a certificate of a rooted triplet inconsistent with $T$ is a conflict that requires only one edition. This will be helpful for further reduction rules.

\begin{lemma} \label{lem:rticonflict}
Let $R_T = (V,\mathcal{R}, T)$ be a system of rooted triplets and $\{a,b,c,d\}$ be a subset of $V$ such that the rooted triplet $t$ over $\{a,b,c\}$ is the unique rooted triplet of $\mathcal{R}[\{a,b,c,d\}]$ inconsistent with $T$. Then $\mathcal{R}[\{a,b,c,d\}]$ is a conflict if and only if $d \in span(t)$.
\end{lemma}

\begin{proof}
We assume w.l.o.g. that $t = bc|a$. This implies that $T_{|\{a,b,c\}}$ is homeomorphic to $ab|c$ or $ac|b$. The two cases are symmetric, so assume the former holds. Suppose that $d\in span(t)$.
By assumption, the rooted triplet over $\{b,c,d\}$ is $t'=T_{|\{b,c,d\}}$. Observe that $d\in span(t)$ implies $t'\neq bc|d$. So there are the two distinct cases for $t'$:
\begin{itemize}
\item $t'=bd|c$: since $T_{\{a,b,c\}}$ is homeomorphic to $ab|c$ by assumption, and since $t$ is the unique rooted triplet of $\mathcal{R}[\{a,b,c,d\}]$ inconsistent with $T$, we have $ad|c \in \mathcal{R}$. Hence the set $\{bc|a, bd|c, ad|c\}$ defines a conflict.
\item $t'=cd|b$: since $T_{\{a,b,c\}}$ is homeomorphic to $ab|c$ by assumption, and since $t$ is the unique rooted triplet of $\mathcal{R}[\{a,b,c,d\}]$ inconsistent with $T$, we have $ab|d \in \mathcal{R}$. Hence the set $\{bc|a, cd|b, ad|b\}$ defines a conflict. 
\end{itemize}
So we proved that if $d \in span(t)$, $\mathcal{R}[\{a,b,c,d\}]$ is a conflict. 
Assume now that $d\notin span(t)$. Again, since  $bc|a$ is the unique rooted triplet of $R_T[\{a,b,c,d\}]$ inconsistent with $T$, every rooted triplet containing $d$ chooses $d$. Thereby $\mathcal{R}[\{a,b,c,d\}]$ is consistent with some binary tree $T'$ and is not a conflict.
\end{proof}


A \emph{tree partition} of $T$ is a set  $\mathcal{P}_T=\{T_1,\dots, T_l\}$ of subtrees defined over pairwise disjoint sets of leaves $\{V_1,\dots V_l\}$ such that for every $i\in [l]$, $T_i$ is the subtree $T_{x_i}$ for some node or leaf $x_i$ of $T$ and $\{V_1,\dots V_l\}$ is a partition of $V$ (see Figure~\ref{fig:treepartition}). For an embedded system of rooted triplets $R_T=(V,\mathcal{R},T)$, a tree partition of $T$ naturally distinguishes two sets of rooted triplets, namely the \emph{external rooted triplets} $\mathcal{R}_E=\{t\in \mathcal{R}\mid \not\!\!\exists\ i\in[l]\ V(t)\subseteq V(T_i)\}$, and the \emph{internal rooted triplets} $\mathcal{R}_I=\mathcal{R}\setminus \mathcal{R}_E$. A subset of external rooted triplets $\mathcal{F} \subseteq \mathcal{R}$ can be \emph{externally certified} whenever it is possible to certify $\mathcal{F}$ with external rooted triplets only. 

\begin{figure}[h]
\centerline{\includegraphics[scale=1.5]{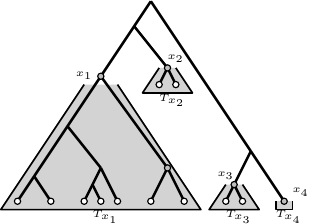}}
\caption{A tree partition $\{T_{x_1}, T_{x_2}, T_{x_3}, T_{x_4}\}$ of an embedded instance of \RTI{}. \label{fig:treepartition}}
\end{figure}

\begin{definition}[Safe Partition] 
\label{def:rtisafe}
Let $R_T = (V,\mathcal{R}, T)$ be an embedded system of rooted triplets. A tree partition $\mathcal{P}_T = \{T_1, \dots, T_l\}$  of $T$ is a \emph{safe partition} if: 
\begin{itemize}
\item $\mathcal{R}_E$ contains at least one triplet inconsistent with $T$ and, 
\item the rooted triplets of $\mathcal{R}_E$ inconsistent with $T$ can be externally certified.
\end{itemize}	
\end{definition}

\begin{polyrule}[Safe Partition] \label{rule:rtisafe}
Let $(R=(V,\mathcal{R}),k)$ be an instance of \RTI{}, $\mathcal{P}_T = \{T_1, \dots, T_l\}$ be a safe partition of $R_T=(V,\mathcal{R}, T)$ where $T$ is a tree defined over $V$. If $\mathcal{F}$ is the set of rooted triplets of $\mathcal{R}_E$ inconsistent with $T$, then return $(R',k')$ where $R'=R\odot_T \mathcal{F}$ and $k'=k-|\mathcal{F}|$.
\end{polyrule}

\begin{lemma} \label{lem:rtisafe}
Rule~\ref{rule:rtisafe} is sound.
\end{lemma}
\begin{proof}
We use the following observation, which follows from the definition of tree partition.

\begin{observation} \label{obs:saferti}
Let  $R_T = (V, \mathcal{R}, T)$ be an embedded system of rooted triplets and  $\mathcal{P}_T = \{T_1, \dots, T_l\}$ be a tree partition of $T$. If  $t$ is a rooted triplet such that $V(t) \subseteq V(T_i)$ for some $1 \leqslant i \leqslant l$, and $d \in V \setminus V(T_i)$, then $d \notin span(t)$.
\end{observation}

 Let us consider the embedded system of rooted triplets $R'_T=(V,\mathcal{R}',T)=R_T\odot_T \mathcal{F}$.

\begin{claim} \label{cl:local-conflict}
If $\mathcal{R}'[C]$, with $C= \{a,b,c,d\}$, is a conflict, then $C\subseteq V(T_i)$ for some $1 \leqslant i \leqslant l$.
\end{claim}
\begin{proofclaim}
Let $t$ be a rooted triplet of $R'_T[C]$ inconsistent with $T$ (suppose that $V(t)=\{a,b,c\}$). Clearly, since $t\notin \mathcal{F}$, there exists $i\in[l]$ such that $\{a,b,c\}\subseteq V(T_i)$. Now, assume for a contradiction that $d$ is not a leaf of $T_i$. By Observation~\ref{obs:saferti}, we have $d\notin span(t$).  Moreover by definition of $R'_T$, every external rooted triplet of $R'_T$ is consistent with $T$. So the unique rooted triplet inconsistent with $T$ in $R'_T[C]$ is $t$. By Lemma~\ref{lem:rticonflict}, this implies that $C$ is not a conflict: contradiction.
\end{proofclaim}

\noindent
Let us now prove that $\mathbf{RTI}(R)\leqslant k$ iff  $\mathbf{RTI}(R')\leqslant k - |\mathcal{F}|$:

\noindent
\smallskip
($\Rightarrow$) 
Let $\mathcal{S}$ be a set of $k$ rooted triplets such that $R\odot \mathcal{S}$ is consistent. Let us define $\mathcal{S}_E=\{t\in \mathcal{S}\mid \not\!\!\exists\ i\in[l]\ V(t)\subseteq V(T_i)\}$ and  $\mathcal{S}_I=\mathcal{S}\setminus\mathcal{S}_E$.
If $R'_T\odot\mathcal{S}_I$ is inconsistent, by  Claim~\ref{cl:local-conflict} there exists  a conflict $C$ contained in $V(T_i)$ for some $i\in[l]$. Clearly none of the editions of $\mathcal{S}_E$ can solve $C$, implying that $R\odot\mathcal{S}$ is also inconsistent: contradiction. So $R'_T\odot\mathcal{S}_I$ is consistent. To conclude if suffices to observe that since $\mathcal{P}_T$ is a safe partition, there exists a set $c(\mathcal{F})$ of triplet-disjoint certificates for $\mathcal{F}$, each of which being  composed only of rooted triplets of $\mathcal{R}_E$. Thereby $|\mathcal{S}_E|\geqslant|\mathcal{F}|$ and $|\mathcal{S}_I|\leqslant k-|\mathcal{F}|$.

\smallskip
\noindent
($\Leftarrow$) Suppose that $\mathcal{S}\subseteq \mathcal{R}'$ is a set of $k - |\mathcal{F}|$ rooted triplets such that $R'\odot\mathcal{S}$ is consistent. Let us recall that $R'=R\odot_T \mathcal{F}$ and let us denote by $\mathcal{F}_T$ the set of rooted triplets such that $R'=R\odot \mathcal{F}_T$. Clearly $R\odot (\mathcal{S}\cup\mathcal{F}_T)$ is consistent, implying that $\textbf{RTI}(R)\leqslant k$.

 \end{proof}


\subsection{Computing a safe partition and kernel size}


To prove that Rule~\ref{rule:rtiuseless} and Rule~\ref{rule:rtisafe} yield a linear vertex kernel, it remains to show that either a safe partition exists and can be computed in polynomial time or prove a (linear) bound on the number of leaves of $R=(V,\mathcal{R})$. Again, we apply the \emph{Conflict Packing} technique. However, defining a conflict packing as a maximal set of certificates that are pairwise rooted triplet-disjoint, will not provide a tight enough lower bound on $\mathbf{RTI}(R)$. To see this consider the following example: 

\smallskip
\noindent
\emph{\textbf{Example:}} Let $R=(V,\mathcal{R})$ be a dense system of rooted triplets such that $C=\{a,b,c,d\}$ and $C'=\{b,c,d,e\}$ are two conflicts of $R$ and $\{ab|c,ac|d,ad|b\}\subset \mathcal{R}$. Observe first that $C$ and $C'$ share the rooted triplet $t$ over $\{b,c,d\}$. Observe also that whichever rooted triplet $t$ is defined over $\{b,c,d\}$, $C$ is a conflict (see Figure~\ref{fig:seed}). It follows that solving $C$ and $C'$ requires at least two editions: one on the rooted triplets of $C'$ and one on $\{ab|c,ac|d,ad|b\}$.

\begin{definition}[Seed] \label{def:seed}
Let $C=\{a,b,c,d\}$ be a conflict in a dense system of rooted triplets $R=(V,\mathcal{R})$. We say that $a$ is a seed for $C$ if $C$ is a conflict in $R\odot \{t\}$ for every rooted triplet $t$ with $V(t)=\{b,c,d\}$.
\end{definition}

In the above example, the leaf $a$ is a seed (see Figure~\ref{fig:seed}).

\begin{figure}[h]
\centerline{\includegraphics[scale=1.75]{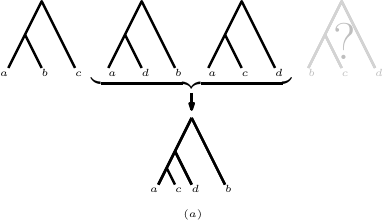}}
\caption{The notion of seed: $a$ is a seed of the conflict $C = \{a,b,c,d\}$ since $C$ is a conflict in $R\odot \{t\}$ for every rooted triplet $t$ with $V(t)=\{b,c,d\}$. \label{fig:seed}} 
\end{figure}	
	
\begin{definition}[Conflict Packing] \label{def:cprti}
Two distinct conflicts of size $4$, $C$ and $C'$, of a dense system of rooted triplets $R=(V,\mathcal{R})$ are \emph{independent} if: 
\vspace{-0.1cm}
\begin{itemize}
\item $|C\cap C'|\leqslant 2$ or 
\vspace{-0.1cm}
\item $|C\cap C'|=3$ and the leaf $\ell\in C\setminus C'$ is a seed for $C$ (or the leaf $\ell'\in C'\setminus C$ is a seed for $C'$)
\end{itemize}
A \emph{conflict packing} is a maximal collection of pairwise independent conflicts of size $4$ of $R$.

\end{definition}

Clearly a conflict packing of a dense system of rooted triplets $R=(V,\mathcal{R})$ can be computed in polynomial time using a greedy algorithm. If $\mathcal{C}$ is a conflict packing, then $V(\mathcal{C})$ denotes the set of leaves covered by the conflicts of $\mathcal{C}$.

\begin{lemma}  \label{lem:cprti}
Let $\mathcal{C}$ be a conflict packing of a dense system of rooted triplets $R = (V, \mathcal{R})$. If $\mathbf{RTI}(R)\leqslant k$, then $|V(\mathcal{C})| \leqslant 4k$.
\end{lemma} 

\begin{proof}
Suppose that $\mathcal{C}$ contains $l$ conflicts $C_1,\dots C_l$. We prove that $\mathbf{RTI}(R)\geqslant l$.
Let $t$ be a rooted triplet such that $C_j$ ($1\leqslant j\leqslant l$) is not a conflict in $R'=R\odot \{t\}$. Let us prove that for every $i\neq j$, $1\leqslant i \leqslant l$, $C_i$ is a conflict of $R'$. 
If $V(t)\not\subseteq C_i$, then $C_i$ is still a conflict in $R'$. So suppose that $V(t)=C_i\cap C_j$, implying that either $\ell_i\in C_i\setminus C_j$ is a seed for $C_i$ or $\ell_j\in C_j\setminus C_i$ is a seed for $C_j$. Observe that the latter case is not possible because $\ell_j$ being a seed implies that editing $t$ cannot solve the conflict $C_j$. Similarly, in the former case, since $\ell_i$ is a seed, $C_i$ remains a conflict in $R'$. It follows that every conflict of $\mathcal{C}$ requires the edition of a distinct rooted triplet, implying that $\mathbf{RTI}(R)\geqslant l$ as claimed. 
As by assumption $\mathbf{RTI}(R)\leqslant k$, we have $|\mathcal{C}|\leqslant k$ and thereby $|V(\mathcal{C}|\leqslant 4k$.
 \end{proof}

\begin{lemma}[Conflict packing]  \label{lem:realcprti}
Let $\mathcal{C}$ be a conflict packing of a dense system of rooted triplets $R = (V, \mathcal{R})$.
There exists an embedded tree $T$ of $R$ such that every rooted triplet $t$ inconsistent with $T$ satisfies $V(t)\subseteq V(\mathcal{C})$. Such a tree can be computed in polynomial time.
\end{lemma}
\begin{proof}
Let $V'=V\setminus V(\mathcal{C})$. Observe that by maximality of $\mathcal{C}$, $R[V']$ is consistent. Let $T'$ be the unique tree with which $R[V']$ is consistent. Moreover, for every leaf $a\in V(\mathcal{C})$, $R_{a}=R[V'\cup\{a\}]$ is consistent (otherwise $\mathcal{C}$ would not be maximal). Thereby there exists a unique binary tree $T_{a}$ such that every rooted triplet $t$ of  $R_{a}$ is consistent with $T_{a}$. In other words $T'$ contains a unique tree edge $e=xz$ which can be subdivided into $xyz$ to attach the leaf $a$ to node $y$. Hereafter the edge $e$ will be called the \emph{locus} of $a$.
The maximality argument on $\mathcal{C}$ also implies that for any pair of leaves $a$ and $b$ in $V(\mathcal{C})$, $R_{ab}=R[V'\cup\{a,b\}]$ is consistent. If $a$ and $b$ have different loci, then $R_{ab}$ is clearly consistent with the tree obtained from $T'$ by inserting $a$ and $b$ in their respective loci. It remains to consider the case where $a$ and $b$ have the same locus.
Let $e=xz$ be a tree edge of $T'$ such that $x$ is the child of $z$ and let $L_e\subseteq V(\mathcal{C})$ be the subset of leaves whose locus is $e$. Given $a,b \in L_e$, we define the binary relation $<_e$ on $L_e$ as follows: 
$$a<_e b \mbox{ if there exists } c\in V' \mbox{ such that } ac|b\in \mathcal{R}$$

%
\begin{claim}
\label{claim:swo}
The relation $<_e$ is a strict weak ordering (i.e. $<_e$ is transitive and asymmetric)
\end{claim}

\begin{proofclaim}
Observe first  that, if $a <_e b$, then the leaf $c$ such that $ac|b\in \mathcal{R}$ belongs to $V(T'_x)$ (otherwise since $R_{ab}$ is consistent we would have $ab|c \in \mathcal{R}$). Assume $<_e$ is not asymmetric. Then there exist two leaves $c\in V(T'_x)$ and $d\in V(T'_x)$ such that $ac|b\in\mathcal{R}$ and $bd|a\in\mathcal{R}$. Since $R_a$ and $R_b$ are consistent and since $c$ and $d$ belong to $V(T'_x)$, we also have $cd|a\in\mathcal{R}$ and $cd|b\in\mathcal{R}$. It follows that $\{a,b,c,d\}$ is a conflict: contradicting the fact that $R_{ab}$ is consistent for every $a,b\in L_e$. So $<_e$ is asymmetric.

Suppose now that we have $a,b,c\in L_e$ such that $a<_e b$ and $b<_e c$. Observe that for every $d\in V(T'_x)$, we have $ad|b\in\mathcal{R}$ and $bd|c\in\mathcal{R}$. 
Assume that $dc|a\in\mathcal{R}$ (the case $ac|d$ is similar). Then whatever the rooted triplet on $\{a,b,c\}$ is, $\{a,b,c,d\}$ is a conflict. Hence $d$ is a seed of the conflict $\{a,b,c,d\}$, contradicting the maximality of $\mathcal{C}$. It follows that $da|c\in \mathcal{R}$ and thereby $a<_e c$. Hence $<_e$ is transitive.
\end{proofclaim}

Let us denote by $\sim_e$ the incomparability relation on $L_e$ with respect to $<_e$. It is well known that as $<_e$ is a strict weak ordering, then $\sim_e$ is an equivalence relation on $L_e$. Moreover the equivalence classes $L_1\dots L_q$ of $\sim_e$ are totally ordered with respect to $<_e$: we say that $L_i<_e L_j$ if for any $a\in L_i$ and $b\in L_j$, $a<_e b$.

Let us now describe how to build the tree $T$ from $T'$ (see Figure~\ref{fig:tree}). For every tree edge $e=xz$ with $x$ a child of $z$ such that $L_e\neq\emptyset$ we proceed as follows. Let $L_1,\dots, L_q$ be the equivalence classes of $\sim_e$ such that $L_i <_e L_j$ for $1\leqslant i<j\leqslant q$. The tree edge $e$ is subdivided into the path $x,y_1,\dots, y_q,y$. For every $i\in[q]$, if $L_i$ contains a unique leaf $a$, then $a$ is attached to node $y_i$. Otherwise, a new node $w_i$ is attached to $y_i$ and we add an arbitrary binary tree (rooted in $w_i$) over the leaves of $L_i$. Clearly, given a conflict packing, $T$ can be computed in polynomial time.

\begin{figure}[h]	
\centerline{\includegraphics[scale=1.75]{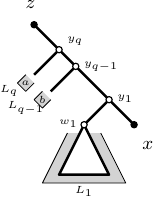}}
\caption{Illustration of the construction of the tree $T'$. \label{fig:tree}}
\end{figure}
	
We now prove that $T$ satisfies the statement. Let $t = \{a,b,c\}$ be any triplet of $\mathcal{R}$. If $V(t) \subseteq V'$, then $t$ is consistent by construction. If $V(t) \cap V(\mathcal{C}) = \{a\}$, then $t$ is consistent with $T$ since $R_a$ is consistent and $a$ has been inserted to its locus. Finally, assume $V(t) \cap V(\mathcal{C}) = \{a,b\}$. If $a$ and $b$ have different loci then $t$ is clearly consistent with $T$. Now, if $a$ and $b$ have the same locus $e$ then $t$ is consistent since $a$ and $b$ have been added to $e$ according to the strict weak ordering $<_e$. It follows that any triplet of $T$ such that $V(t) \cap V' \neq \emptyset$ is consistent with $T$.
 \end{proof}

%
%
\begin{theorem} \label{thm:kernelrti}
The \RTI{} problem admits a kernel with at most $5k$ leaves.
\end{theorem}

\begin{proof}
We prove the result by showing that any instance $(R = (V, \mathcal{R}), k)$ of \RTI{} such that $R$ is reduced under Rule~\ref{rule:rtiuseless} and $|V| > 5k$ is either a negative instance or can be reduced using Rule~\ref{rule:rtisafe}. \\ 

We start by computing a conflict packing $\mathcal{C}$ in polynomial time, for instance using a greedy algorithm. Let $T$ be the tree obtained in polynomial time through Lemma~\ref{lem:realcprti}. Construct in polynomial time the following bipartite graph $B = (I \cup V', E)$ where: 
\begin{itemize}
\vspace{-0.2cm}
\item $V' = V \setminus V(\mathcal{C})$, 
\vspace{-0.2cm}
\item there is a vertex $i_{t}$ in $I$ for every rooted triplet $t$ inconsistent with $T$ and, 
\vspace{-0.2cm}
\item $i_{t}v \in E$ if $v \in V'$ and $\{v\}\cup V(t)$ is a certificate of $t$.
\end{itemize}

\begin{figure}[t]
\begin{center}
\includegraphics[scale=1.5]{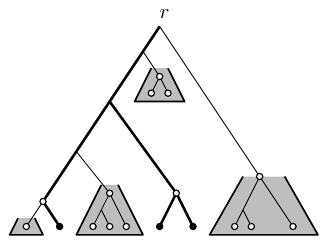}
\end{center}
\caption{Illustration of the construction of $\mathcal{P}_T$. The black vertices belong to $V' \setminus D_{V'}$ and the bold edges represent $S$. The partition $\mathcal{P}_T$ is pictured by the filled sets.}
\label{fig:saferti}
\end{figure}

Observe that if $B$ has a matching of size $k+1$, then $\mathbf{RTI}(R)> k$ (Lemma~\ref{lem:cprti}) and thus the instance is negative. Hence a minimum vertex cover $D$ of $B$ has size at most $k$~\cite{BM76}. We denote $D_I=D\cap I$ and $D_{V'}=D\cap V'$. 
Assume that $|V| > 5k$. Since $|D_{V'}| \leqslant k$ and $|V(\mathcal{C})| \leqslant 4k$ (Lemma~\ref{lem:cprti}), we have $V' \setminus D_{V'} \neq \emptyset$. 
Let $\mathcal{P}_T=\{T_1,\dots, T_l\}$ be a tree partition of $R_T$ such that every tree $T_i$, $i\in[l]$, consists of either a leaf of $V'\setminus D_{V'}$ or a connected component of $T \setminus S$, where $S$ is the smallest spanning subtree of $(V'\setminus D_{V'}) \cup \{r\}$ ($r$ being the root of $T$, see Figure~\ref{fig:saferti}).

\begin{claim} \label{claim:pissaferti}
The partition $\mathcal{P}_T = \{T_1, \dots, T_l\}$ is safe.
\end{claim}

\begin{proofclaim}	
Let $a$ be any leaf of $V' \setminus D_{V'}$. By Lemma~\ref{lem:realcprti}, $a$ is not contained in any rooted triplet inconsistent with $T$. Since $R$ is reduced under Rule~\ref{rule:rtiuseless}, there exists an inconsistent rooted triplet $t$ such that $a\in span(t)$. It follows that $\mathcal{R}_E$ contains at least one inconsistent rooted triplet.
Let $t\in \mathcal{R}_E$ be a rooted triplet inconsistent with $T$. By construction of $\mathcal{P}_T$, there exists a leaf $a \in (V' \setminus D_{V'})\cap span(t)$. Then $\{a\}\cup V(t)$ is a certificate of $t$ and $i_ta$ is an edge of $B$. Observe that since $D$ is a vertex cover and $a\notin D_{V'}$, the vertex $i_{t}$ has to belong to $D_I$ to cover the edge $i_{t}a$. Thereby the subset $I'\subseteq I$ corresponding to the rooted triplets of $\mathcal{R}_E$ inconsistent with $T$ is included in $D_I$.  \\

By Lemma~\ref{lem:matching}, we know that $I'$ can be matched into $V' \setminus D_{V'}$ in $B$. 
Since every leaf of $V' \setminus D_{V'}$ is a singleton in $\mathcal{P}_T$, the existence of the matching shows that the set of rooted triplets of $\mathcal{R}_E$ inconsistent with $T$ can be certified using rooted triplet of $\mathcal{R}_E$ only, and hence $\mathcal{P}_T$ is safe. 
\end{proofclaim}

This concludes the proof.
 \end{proof}

\section{Linear vertex-kernel for {\sc Betweenness Triple Inconsistency}}
\label{sec:bit}


Let $V$ be a set of vertices. A \emph{system of betweenness constraints over $V$} is a pair $R=(V,\mathcal{B})$ where $\mathcal{B}$ is a set of betweenness contraints. Let us recall that a betweenness constraint is represented by an ordered triple $t=(a,b,c)$ of distinct vertices and note $V(t)=\{a,b,c\}$. We say that $R=(V,\mathcal{B})$ is \emph{dense} if $\mathcal{B}$ contains a betweenness constraint for every triple of distinct vertices $(a,b,c)$. Given $S \subseteq V$, we define $R[S]$ as the pair $(S,\mathcal{B}[S])$, where $\mathcal{B}[S]=\{t\in \mathcal{B}\mid V(t)\subseteq S\}$.


The betweenness constraint $t=(a,b,c)$ is \emph{consistent} with an ordering $\sigma$ of the vertices of $V$ if $b$ lies between $a$ and $c$ in $\sigma$, that is $a<_{\sigma} b<_{\sigma} c$ or  $c<_{\sigma} b<_{\sigma} b$. Observe that $(a,b,c)$ and $(c,b,a)$ are the same betweenness constraints. We say that $R$ is \emph{consistent} if there exists an ordering $\sigma$ of $V$ such that every betweenness constraint $t\in\mathcal{B}$ is consistent with $\sigma$. We also say that $R$ is \emph{consistent with} $\sigma$. Observe that if $R$ is dense and consistent, then there exists a unique ordering $\sigma$ (up to reversal) with which $R$ is consistent.
If such an ordering does not exist, then $R$ is \emph{inconsistent}.  A \emph{conflict} is a subset $C\subseteq V$ such that $R[C]$ is inconsistent. 

\begin{lemma} \label{lem:consistentbit}
A dense system of betweenness triplets $R = (V, \mathcal{B})$ is consistent if and only if it does not contain any conflict on $4$ vertices.
\end{lemma}
\begin{proof}
The forward direction is trivial. Let us prove by induction on the number of vertices that  if $R$ does not contain any conflict on $4$ vertices, then $R$ is consistent. The statement obviously holds for $n=4$. Suppose it is true up to $n$ vertices and let 
$R=(V,\mathcal{B})$ such that $|V|=n+1$. Let $x\in V$, then by assumption $R'=R[V\setminus \{x\}]$ is consistent. Let $\sigma'=v_1\dots v_i\dots v_n$ be the unique ordering with which $R'$ is consistent.

\begin{claim} \label{cl:bit1}
Let $i,j$ be integers such that $1\leqslant i<j\leqslant n$, $(v_i,x,v_j)\in\mathcal{B}$ and $j-i$ is minimum. If $j\neq i+1$ then $R$ contains a conflict on $4$ vertices.
\end{claim}
\begin{proofclaim}
If $j\neq i+1$, there exists $l$ such that $i<\ell<j$. 
As $R'$ is consistent with $\sigma'$, $(v_i,v_{\ell},v_j)\in\mathcal{B}$. Together with the assumption that $(v_i,x,v_j)\in\mathcal{B}$, this implies that either $(v_i,x,v_{\ell})\in \mathcal{B}$ or $(v_{\ell},x,v_j)\in \mathcal{B}$ (otherwise $R[\{v_{\ell},v_i,v_j,x\}]$ would be a conflict of size $4$). But then $\ell-i<j-i$ contradicts the choice of $v_i$ and $v_j$. 
\end{proofclaim}

\begin{claim} \label{cl:bit2}
Let $i,j$ be integers such that $1\leqslant i<j\leqslant n$. If $(v_i,x,v_j)\in\mathcal{B}$, then for every $\ell$ and $\ell'$ such that $1\leqslant \ell\leqslant i$ and  $j\leqslant \ell'\leqslant n$, we have $(v_{\ell},x,v_j)\in \mathcal{B}$, $(v_i,x,v_{\ell'})\in \mathcal{B}$ and $(v_{\ell},x,v_{\ell'})\in\mathcal{B}$.
\end{claim}
\begin{proofclaim}
Let us first prove that $(v_{\ell},x,v_j)\in \mathcal{B}$. 
As $R'$ is consistent with $\sigma'$, $(v_{\ell},v_i,v_j)\in\mathcal{B}$. Together with the assumption that $(v_i,x,v_j)\in\mathcal{B}$, this implies that $(v_{\ell},x,v_j)\in\mathcal{B}$ (otherwise $R[\{v_{\ell},v_i,v_j,x\}]$ would be a conflict of size $4$). The proof that $(v_i,x,v_{\ell'})\in \mathcal{B}$  is symmetric. Knowing that $(v_i,x,v_{\ell'})\in \mathcal{B}$, as $\ell<i$ we can apply the above arguments on $\ell$ and $i,\ell'$ to prove that that $(v_{\ell},x,v_{\ell'})\in\mathcal{B}$.
\end{proofclaim}

\begin{claim} \label{cl:bit3}
Let $i,j$ be integers such that $1\leqslant i<j\leqslant n$.  If $(v_i,v_j,x)\in\mathcal{B}$, then for every $\ell$ such that $\ell\leqslant i$, we have $(v_{\ell},v_j,x)\in \mathcal{B}$ and $(v_{\ell},v_i,x)\in \mathcal{B}$.
If $(x,v_i,v_j)\in\mathcal{B}$, then for every $\ell'$ such that $j\leqslant \ell'$, we have $(x,v_i,v_{\ell'})\in \mathcal{B}$ and $(x,v_j,v_{\ell'})\in \mathcal{B}$. 
\end{claim}
\begin{proofclaim}
As $R'$ is consistent with $\sigma'$, $(v_{\ell},v_i,v_j)\in\mathcal{B}$. Together with the assumption that $(v_i,v_j,x)\in\mathcal{B}$, this implies that $(v_{\ell},v_j,x)\in\mathcal{B}$ and $(v_{\ell},v_i,x)\in \mathcal{B}$ (otherwise $R[\{v_{\ell},v_i,v_j,x\}]$ would be a conflict of size $4$). By symmetric arguments, $(x,v_i,v_j)\in\mathcal{B}$ implies $(x,v_i,v_{\ell'})\in \mathcal{B}$ and $(x,v_j,v_{\ell'})\in \mathcal{B}$.
\end{proofclaim}

So by Claim~\ref{cl:bit1}, if $\mathcal{B}$ contains a triplet of the form $(v_{\ell},x,v_{\ell'})$, then there exists $i$, $1\leqslant i<n$, such that $(v_i,x,v_{i+1})\in \mathcal{B}$. By Claim~\ref{cl:bit2}, for every $\ell$ and $\ell'$ such that $\ell<i$ and $i+1<\ell'$ we have: $(v_{\ell},x,v_{i+1})\in \mathcal{B}$, $(v_i,x,v_{\ell'})\in \mathcal{B}$ and $(v_{\ell},x,v_{\ell'})\in\mathcal{B}$. By Claim~\ref{cl:bit3}, for every $1\leqslant \ell<\ell'\leqslant i$, $(v_{\ell},v_{\ell'},x)\in\mathcal{B}$ and for every $i+1\leqslant h<h'\leqslant n$, $(x,v_h,v_{h'})\in\mathcal{B}$. It follows that inserting $x$ in $\sigma'$ between $v_i$ and $v_{i+1}$ yields an ordering $\sigma$ with which $R$ is consistent. It remains to consider the case where for every pair $\ell, \ell'$, we have $(x,v_{\ell},v_{\ell'})\in\mathcal{B}$ or $(x,v_{\ell}',v_{\ell})\in\mathcal{B}$. By Claim~\ref{cl:bit3}, if $(x,v_1,v_2)\in\mathcal{B}$, then for every $1\leqslant i<j\leqslant n$, $(x,v_i,v_j)\in\mathcal{B}$ and thereby $x$ can be inserted before $v_1$ in $\sigma'$. Likewise, if $(v_{n-1},v_n,x)\in\mathcal{B}$, then by Claim~\ref{cl:bit3} for every $1\leqslant i<j\leqslant n$, $(v_i,v_j,x)\in\mathcal{B}$ and thereby $x$ can be inserted after $v_n$ in $\sigma'$. In any case, $R$ is consistent with the resulting ordering $\sigma$.
\end{proof}

\emph{Editing} a betweenness constraint $t\in\mathcal{B}$ into $t'$ consists of replacing $t'$ by $t$ in $\mathcal{B}$. Likewise if $\mathcal{S}$ is a set of betweenness constraints, \emph{editing $R$ according to $\mathcal{S}$} consists for every $t'\in\mathcal{S}$ in editing $t\in\mathcal{B}$, such that $V(t)=V(t')$, into $t'$.  The resulting system of betweenness constraints is  denoted by $R\odot \mathcal{S}$. If $\mathcal{S}\subseteq \mathcal{B}$, \emph{editing $\mathcal{S}$ according to the vertex ordering $\sigma$} consists in editing every betweenness constraint $t\in\mathcal{S}$ so that it becomes consistent with $\sigma$. In other words, if $V(t)=\{a,b,c\}$, then $t$ is edited into $t'=(a,b,c)$ if and only if $a<_{\sigma} b <_{\sigma} c$. This latter operation is denoted $R\odot_{\sigma}\mathcal{S}$.
We denote by $\mathbf{BTI}(R)$ the minimum size of a set $\mathcal{S}$ of betweenness constraints such that $R\odot\mathcal{S}$ is consistent.

An \emph{ordered system of betweenness constraints} over $V$ is a triple $R_{\sigma}=(V,\mathcal{B},\sigma)$ such that $(V,\mathcal{B})$ is dense and $\sigma$ is an ordering of $V$. For $S\subseteq V$, we also define $R_\sigma[S]=(S,\mathcal{B}[S],\sigma_S)$ with $\sigma_S$ being the restriction of $\sigma$ to $S$. As shown by the following lemma, the vertex ordering $\sigma$ eases the identification of conflict of size $4$.

\begin{lemma}\label{lem:conflictbit}
Let $R_{\sigma} = (V, \mathcal{B},\sigma)$ be an ordered dense system of betweenness triplets. If $R_\sigma[\{a,b,c,d\}]$, with $\{a,b,c,d\}\subseteq V$, contains a unique betweenness constraint $t$ inconsistent with $\sigma$, then $\{a,b,c,d\}$ is a conflict.
\end{lemma}
\begin{proof}
Assume w.l.o.g. that $V(t) = \{a,b,c\}$ and $t = (b,c,a)\in \mathcal{B}$. It follows that either $a <_\sigma b <_\sigma c$ or  $b <_\sigma a <_\sigma c$ (the cases of reverse orderings are symmetric). Let us denote by $t'$ and $t"$ the betweenness constraints such that $V(t')=\{b,c,d\}$ and $V(t")=\{a,c,d\}$. Let us show that the three possibilities for $t'$ lead to a conflict on $\{a,b,c,d\}$:
\begin{itemize}
\item $t'=(b,c,d)$: If $a <_\sigma b <_\sigma c$, then $t'$ being consistent with $\sigma$ implies that $a <_\sigma b <_\sigma c<_{\sigma} d$. As $t$ is the only betweenness constraint of $R_\sigma[\{a,b,c,d\}]$ inconsistent with $\sigma$, we have that $t"=(a,c,d)\in\mathcal{B}$. Likewise, if $b <_\sigma a <_\sigma c$ holds, then we have $b <_\sigma a <_\sigma c<_{\sigma} d$ and $t"=(a,c,d)\in\mathcal{B}$. 
Observe that in both cases, $\{t,t',acd\}\subset\mathcal{B}$ implies a conflict on $\{a,c,b,d\}$. 
\item $t'=(c,b,d)$: If $a <_\sigma b <_\sigma c$, then $t'$ being consistent with $\sigma$ implies that $d <_\sigma b <_\sigma c$. As $t$ is the only betweenness constraint of $R_\sigma[\{a,b,c,d\}]$ inconsistent with $\sigma$, either $t"=(d,a,c)\in\mathcal{B}$ or $t"=(a,d,c)\in\mathcal{B}$. Likewise, if $b <_\sigma a <_\sigma c$ holds, then $t"=(d,a,c)\in\mathcal{B}$. 
Observe that in both cases, $\{t,t',t"\}\subset\mathcal{B}$ implies a conflict on $\{a,c,b,d\}$. 
\item $t'=(c,d,b)$: If $a <_\sigma b <_\sigma c$, then $t'$ being consistent with $\sigma$ implies that $a<_{\sigma} b<_{\sigma} d<_{\sigma} c$. As $t$ is the only betweenness constraint of $R_\sigma[\{a,b,c,d\}]$ inconsistent with $\sigma$, we have $t"=(a,d,c)\in\mathcal{B}$. Likewise, if $b <_\sigma a <_\sigma c$ holds, then $t"=(d,a,c)\in\mathcal{B}$ or $t"=(a,d,c)\in\mathcal{B}$.  
Observe that in both cases, $\{t,t',t"\}\subset\mathcal{B}$ implies a conflict on $\{a,c,b,d\}$.
\end{itemize}
\end{proof}

\subsection{Reduction rule and conflict packing}

The notion of a \emph{sunflower} is classic in the context of kernelization~\cite{Abu10,ALS09,BP11,GM10}. In our context, a \emph{sunflower} $\mathcal{S}$ is defined a set of conflicts $\{C_1, \dots, C_m\}$ pairwise intersecting in exactly one triplet $t$, called the \emph{centre} of $\mathcal{S}$. Clearly the existence of a sunflower of size $m\geqslant k+1$ forces the edition of its centre. Observe that a betweenness constraint can be edited in two different manners. To avoid this drawback, we use a refined definition of a sunflower in the context of ordered system of betweenness constraints. As shown in~\cite{Per13}, this leads to a reduction rule which will be the unique reduction rule of our linear kernel.

\begin{definition}[Simple sunflower~\cite{Per13}] 
Let $\mathcal{S} = \{C_1, \dots, C_m\}$ be a sunflower of an ordered system of betweenness constraints $R_\sigma = (V, \mathcal{B}, \sigma)$. Then $\mathcal{S}$ is \emph{simple} if for every  $1 \leqslant i \leqslant m$, its centre $t$ is the unique betweenness constraint of $R_\sigma[C_i]$ inconsistent with $\sigma$. 
\end{definition}

\begin{polyrule}[Simple sunflower rule \cite{Per13}] \label{rule:simplesun} 
Let $(R=(V,\mathcal{B}),k)$ be an instance of \BIT{} and $\sigma$ be a vertex ordering of $V$. If $\mathcal{S} = \{C_1, \dots, C_m\}$, with $m > k$, is a simple sunflower of centre $t$ of $R_\sigma = (V, \mathcal{B}, \sigma)$, then return $(R',k-1)$ where $R'=R\odot_{\sigma}\{t\}$. 
\end{polyrule}

As in the context of \RTI{}, a conflict packing has to allow a large overlapping of the conflicts in presence of a seed.

\begin{definition}[Seed]
Let $C=\{a,b,c,d\}$ be a conflict of a dense system of betweenness constraints $R=(V,\mathcal{B})$. A vertex $a$ is a \emph{seed} for $C$ if $C$ is a conflict in $R\odot \{t\}$ for every possible betweenness constraint $t$ with $V(t)=\{b,c,d\}$.
\end{definition}

\begin{definition}[Conflict Packing]
Two distinct conflicts $C$ and $C'$ of a dense system of betweenness constraints $R=(V,\mathcal{B})$ are \emph{independent} if:
\vspace{-0.1cm}
\begin{itemize}
\item $|C\cap C'|\leqslant 2$, or
\vspace{-0.1cm}
\item $|C\cap C'|=3$ and the leaf $\ell\in C\setminus C'$ is a seed for $C$ (or the leaf $\ell'\in C'\setminus C$ is a seed for $C'$)
\end{itemize}
A \emph{conflict packing} is a maximal collection of pairwise independent conflicts of $R$.
\end{definition}

\begin{lemma}  \label{lem:cpbit}
Let $\mathcal{C}$ be a conflict packing of a dense system of betweenness constraints $R = (V, \mathcal{B})$. If $\mathbf{BTI}(R)\leqslant k$, then $|V(\mathcal{C})|\leqslant 4k$.
\end{lemma} 
\begin{proof}
Suppose that $\mathcal{C}$ contains $l$ conflicts $C_1,\dots C_l$. We prove that $\mathbf{BTI}(R)\geqslant l$.
Let $t$ be a betweenness constraint such $C_j$ ($1\leqslant j\leqslant l$) is not a conflict in $R'=R\odot \{t\}$. Let us prove that for every $i\neq j$, $1\leqslant i \leqslant l$, $C_i$ is a conflict of $R'$. 
If $V(t)\not\subseteq C_i$, then $C_i$ is still a conflict in $R'$. So suppose that $V(t)=C_i\cap C_j$, implying that either $v_i\in C_i\setminus C_j$ is a seed for $C_i$ or $v_j\in C_j\setminus C_i$ is a seed for $C_j$. Observe that the latter case is not possible because $v_j$ being a seed implies that editing $t$ cannot solve the conflict $C_j$. Similarly, in the former case, since $v_i$ is a seed, $C_i$ remains a conflict in $R'$. It follows that every conflict of $\mathcal{C}$ requires the edition of a distinct betweenness constraint, implying that $\mathbf{BTI}(R)\geqslant l$ as claimed. 
As by assumption $\mathbf{BTI}(R)\leqslant k$, we have $|\mathcal{C}|\leqslant k$ and thereby $|V(\mathcal{C}|\leqslant 4k$.
\end{proof}

\begin{lemma} \label{lem:realcpbit}
Let $\mathcal{C}$ be a conflict packing of a dense system of betweenness constraints $R = (V, \mathcal{B})$. There exists an ordering $\sigma$ of $V$ such that every betweenness triplet $t$ inconsistent with $\sigma$ satisfies $V(t) \subseteq V(\mathcal{C})$. Such an ordering can be computed in polynomial time.
\end{lemma}
\begin{proof}
Let $V' = V \setminus V(\mathcal{C})$. Observe that $R' = R[V']$ is consistent with a (unique) ordering $\sigma'$. Moreover, notice that for every vertex $a \in V(\mathcal{C})$ the instance $R_a = R[V' \cup \{a\}]$ is consistent (otherwise $\mathcal{C}$ would not be maximal). Thereby there exists a unique ordering $\sigma'_a$ such that every triplet of $R_a$ is consistent with $\sigma'_a$. In other words, $\sigma'$ contains a unique pair of consecutive vertices $(u,w)$ such that we have $u <_{\sigma'_a} a <_{\sigma'_a} w$. We call such a pair the \emph{locus} of $a$. Using again the maximality argument on $\mathcal{C}$, we know that for any pair of vertices $a,b \in V(\mathcal{C})$ the instance $R_{ab} = R[V' \cup \{a,b\}]$ is consistent. Hence if $a$ and $b$ have different loci then $R_{ab}$ is clearly consistent with the ordering obtained from $\sigma'$ by inserting $a$ and $b$ in their respective loci. 
It remains to consider the case where $a$ and $b$ have the same locus. Given two consecutive vertices $u$ and $w$ of $\sigma'$, let $L_{uw} \subseteq V(\mathcal{C})$ be the subset of vertices of $V(\mathcal{C})$ whose locus is $uw$. Given $a,b \in L_{uw}$ we define the binary relation $<_{uw}$ on $L_{uw}$ as follows: 
$$a <_{uw} b \mbox{ iff } (a,b,w) \in \mathcal{B}$$
	
\begin{claim} \label{claim:orderbit}
The relation $<_{uw}$ is a strict total ordering.
\end{claim}
\begin{proofclaim}
Observe that $<_{uw}$ is asymmetric by definition. Hence it remains to prove that $<_{uw}$ is transitive.
Suppose that $a,b,c \in L_{uw}$ are such that $a <_{uw} b$ and $b <_{uw} c$. This means that $(a,b,w) \in \mathcal{B}$ and $(b,c,w) \in \mathcal{B}$. Now assume that $(a,c,w) \notin \mathcal{B}$, \ie w.l.o.g. that $(w,a,c) \in \mathcal{B}$. Then $\{a,b,c,w\}$ is a conflict for any choice of $\{a,b,c\}$. Hence $w$ is a seed of the conflict $\{a,b,c,w\}$, meaning that the conflict packing $\mathcal{C}$ is not maximal: contradiction. It follows that $(a,c,w) \in \mathcal{B}$ and thereby $a <_{uw} c$. Hence $<_{uw}$ is transitive.
\end{proofclaim}
	
The ordering $\sigma$ is built from $\sigma'$ as follows: let $u$ and $w$ be two consecutive vertices $u$ and $w$ such that $L_{uw} \neq \emptyset$, then insert the vertices of $L_{uw}$ according to $<_{uw}$. Such an ordering can clearly be built in polynomial time. Let us prove that every betweenness constraint $t\in\mathcal{B}$ 
such that $V(t) \cap V' \neq \emptyset$ is consistent with $\sigma$:
\vspace{-0.1cm}
\begin{itemize}
\item $V(t) \subseteq V'$: since $t$ is consistent with $\sigma'$ and $\sigma$ is an extension of $\sigma'$, $t$ is consistent with $\sigma$. 
\vspace{-0.1cm}
\item $|V(t) \cap V'| = 2$ with $a\in V(t)\setminus V'$: snce $R_a$ is consistent and $a$ is inserted at its locus, $t$ is consistent with $\sigma$. 
\vspace{-0.1cm}
\item $|V(t) \cap V'| = 1$ with $a,b \in V(t)\setminus V'$: two cases may occur. If $a$ and $b$ have different loci then $t$ is consistent with $\sigma$ since $R_{ab}$ is consistent. In the other case, we know that $t$ is consistent since $a$ and $b$ have been inserted in their common locus according to $<_{uw}$.
\end{itemize}
\end{proof}

\subsection{Kernel size}

We can now prove that the simple sunflower reduction rule leads to a linear kernel.

\begin{theorem} \label{thm:bit}
The \BIT{} problem admits a kernel with $5k$ vertices.
\end{theorem}
\begin{proof}
Let $ R = (V, \mathcal{B})$ be an instance of \BIT{}. 
We prove that if $|V|> 5k$, then in polynomial time we can either compute an edition set $\mathcal{S}$ of size at most $k$ such that $R\odot\mathcal{S}$ is consistent, or decide that such a set does not exist. To that aim, we first compute a conflict packing  $\mathcal{C}$ of $R$ and then an ordering $\sigma$ satisfying the conditions of Lemma~\ref{lem:realcpbit}. By Lemma~\ref{lem:cpbit}, if $\mathbf{BTI}(R)\leqslant k$, then $|V(\mathcal{C})| \leqslant 4k$. Hence $V'= V \setminus V(\mathcal{C})$ contains at least $k + 1$ vertices that do not belong to any triplet inconsistent with $\sigma$.

\begin{claim} \label{claim:sun}
$R_{\sigma}=(V,\mathcal{B},\sigma)$ contains a simple sunflower $\mathcal{S} = \{C_1, \dots, C_m\}$ with $m > k$.
\end{claim}
\begin{proofclaim}
Let $t$ be any betweenness constraint inconsistent with $\sigma$ and $V'' =\{d_1, \dots, d_{k+1}\}$ be a subset of $k + 1$ vertices of $V'$. By Lemma~\ref{lem:conflictbit}, for every $d_i\in V"$, $C_{d_i} = V(t) \cup \{d_i\}$ is a conflict. 
By construction of $\sigma$ (see Lemma~\ref{lem:realcpbit}), $t$ is the unique betweenness constraint of $C_{d_i}$ inconsistent with $\sigma$. Hence $\mathcal{S} = \{C_1, \dots, C_{k+1}\}$ is a simple sunflower of center $t$.
\end{proofclaim}
	
Using Rule~\ref{rule:simplesun} we have to edit the center $t \in \mathcal{B}$ of $\mathcal{S}$ according to $\sigma$. Observe that $R\odot_{\sigma}\{t\}$ still contains at least $k + 1$ vertices that do not belong to any betweenness constraint inconsistent with $\sigma$. It follows that the whole set of betweenness constraints inconsistent with $\sigma$ must be edited. Hence if $R$ contains less than $k$  betweenness constraints inconsistent with $\sigma$, then we return a small trivial \textsc{Yes}-instance; otherwise we return a small trivial \textsc{No}-instance.
 \end{proof}

\section{Conclusion}

In this paper we develop the Conflict Packing technique, which allows us to improve several kernelization results for editing problems on dense instances in an unifying manner. The main result we obtain is a linear vertex-kernel for \RTI{}, improving the previous bound of $O(k^2)$ vertices~\cite{GM10}. We also use Conflict Packing to prove that the well-known \FAST{} problem admits a linear vertex-kernel. While such a result was already known to exist~\cite{BFG+11}, we believe that our proof is simpler. Morever, Conflict Packing can be used to obtain a quadratic vertex-kernel for the \FASBT{} problem in a similar way. Here again, such a result was already known to exist~\cite{GX13}, using the concept of so-called bimodules. We however find interesting to provide an unified framework to cope with such problems. Finally, we provide a linear vertex-kernel for the \BIT{} problem, the size of which has been recently improved using some of the results we present here~\cite{Per13}. It is worth to notice that all aforementioned problems share similar properties (besides being defined on dense instances). We would also like to mention that unlike the known linear vertex-kernels for \FAST{}~\cite{BFG+11} and \BIT{}~\cite{Per13}, our algorithms do not use any constant-factor approximation algorithm as a subroutine. It would thus be very interesting to determine whether the Conflict Packing technique can be applied to problems known not to admit any constant-factor approximation algorithm. \\

	
We conclude by addressing some open problems. First, a natural question that arises from the result on \FASBT{} is whether a \emph{linear} vertex-kernel can be achieved for this problem? Another important question is whether the \RTI{} problem admits a constant-factor approximation algorithm. Such an algorithm together with the safe partition reduction rule would also imply a linear vertex-kernel for the problem. Moreover, there exists a lot of problems that remain $NP$-hard on dense instances (see e.g.~\cite{AA07,GM77}): we believe that our technique will yield linear vertex-kernels for a number of these problems as well. Finally, it would be very interesting to determine whether a more general framework can be built on the Conflict Packing technique, providing some meta-theorem on kernelization algorithms.

\paragraph{Acknowledgments.} Research supported by the AGAPE project (ANR-09-BLAN-0159) and the projet KERNEL from the Languedoc-Roussillon program \emph{Chercheur d'avenir}.

\bibliographystyle{plain}
\bibliography{FAST-RTI}

\end{document}